\documentclass[11pt]{article}
\usepackage[letterpaper,margin=1.00in]{geometry}
\usepackage{amsmath, amssymb, amsthm, thmtools, amsfonts, dsfont}
\usepackage{bbm}

\usepackage{ifthen}

\usepackage{tikz}
\usetikzlibrary{positioning,decorations.pathreplacing,calligraphy}

\usepackage{cite}
\usepackage{appendix}
\usepackage{graphicx}
\usepackage{color}
\usepackage{algorithm}
\usepackage[noend]{algpseudocode}
\usepackage{epstopdf}
\usepackage[textsize=tiny]{todonotes}
\usepackage{soul}

\usepackage{framed}
\usepackage[framemethod=tikz]{mdframed}
\usepackage[bottom]{footmisc}
\usepackage[shortlabels]{enumitem}
\setitemize{noitemsep,topsep=3pt,parsep=3pt,partopsep=3pt}
\usepackage[font=small]{caption}
\usepackage{xspace}

\usepackage{mathtools}

\newtheorem{theorem}{Theorem}[section]
\newtheorem{lemma}[theorem]{Lemma}
\newtheorem{meta-theorem}[theorem]{Meta-Theorem}

\newtheorem{corollary}[theorem]{Corollary}

\newtheorem{definition}[theorem]{Definition}

\newcommand{\FullOrShort}{full}

\ifthenelse{\equal{\FullOrShort}{full}}{
    
  \newcommand{\fullOnly}[1]{#1}
  \newcommand{\shortOnly}[1]{}
  }{
    \newcommand{\fullOnly}[1]{}
    \newcommand{\shortOnly}[1]{#1}
  }

\definecolor{darkgreen}{rgb}{0,0.5,0}
\definecolor{dkblue}{rgb}{0,0,0.6}
\usepackage{hyperref}
\hypersetup{
    unicode=false,          % non-Latin characters in Acrobat’s bookmarks
    colorlinks=true,        % false: boxed links; true: colored links
    linkcolor=dkblue,          % color of internal links (change box color with linkbordercolor)
    citecolor=darkgreen,        % color of links to bibliography
    filecolor=magenta,      % color of file links
    urlcolor=cyan           % color of external links
}
\usepackage[capitalize, nameinlink]{cleveref}
\crefname{theorem}{Theorem}{Theorems}
\Crefname{lemma}{Lemma}{Lemmas}
\Crefname{observation}{Observation}{Observations}
\Crefname{equation}{}{}

\algnewcommand\algorithmicswitch{\textbf{switch}}
\algnewcommand\algorithmiccase{\textbf{case}}

\algdef{SE}[SWITCH]{Switch}{EndSwitch}[1]{\algorithmicswitch\ #1\ \algorithmicdo}{\algorithmicend\ \algorithmicswitch}%
\algdef{SE}[CASE]{Case}{EndCase}[1]{\algorithmiccase\ #1}{\algorithmicend\ \algorithmiccase}%
\algtext*{EndSwitch}%
\algtext*{EndCase}%
\usepackage{thmtools} 
\usepackage{thm-restate}

\usepackage[textsize=tiny]{todonotes}

\newcommand{\eps}{\varepsilon}

\newcommand{\poly}{{\rm poly}}

\providecommand{\E}{{\rm \mathbb{E}}}

\providecommand{\poly}{{\rm poly}}

\newcommand{\set}[1]{\left\{#1\right\}}

\renewcommand{\paragraph}[1]{\medskip\noindent {\bf #1}}

\usepackage{setspace}

\algdef{SE}[SWITCH]{Switch}{EndSwitch}[1]{\algorithmicswitch\ #1\ \algorithmicdo}{\algorithmicend\ \algorithmicswitch}%
\algdef{SE}[CASE]{Case}{EndCase}[1]{\algorithmiccase\ #1}{\algorithmicend\ \algorithmiccase}%
\algtext*{EndSwitch}%
\algtext*{EndCase}%
\newcommand{\CONGEST}{$\mathsf{CONGEST}$\xspace}
\newcommand{\LOCAL}{$\mathsf{LOCAL}$\xspace}

\renewcommand{\tilde}{\widetilde}

\newcommand{\red}{\mathsf{R}}
\newcommand{\blue}{\mathsf{B}}
\newcommand{\dred}{d_{\red}}
\newcommand{\dblue}{d_{\blue}}
\newcommand{\typeO}{\mathsf{O}}
\newcommand{\typeC}{\mathsf{C}}

\newcommand{\TSO}{\ensuremath{\mathsf{T_{SO}}}\xspace}
\newcommand{\TBO}{\ensuremath{\mathsf{T_{BO}}}\xspace}
\newcommand{\TLLL}{\ensuremath{\mathsf{T_{LLL}}}\xspace}

\renewcommand{\vec}[1]{\boldsymbol{#1}}
\usepackage{mathtools}
\DeclarePairedDelimiter{\abs}{\lvert}{\rvert}

\makeatletter
\let\oldabs\abs
\def\abs{\@ifstar{\oldabs}{\oldabs*}}
\makeatother

\date{}

\title{\huge On the Complexity of Distributed Edge Coloring and Orientation Problems}

\author{
    \begin{minipage}{0.25\textwidth}
        \centering
        Sebastian Brandt\\
        \small CISPA Helmholtz Center for Information Security\\
         \small Saarbrücken, Germany\\
        \small \href{mailto:brandt@cispa.de}{\color{black}brandt@cispa.de}
    \end{minipage}
    \and
    \begin{minipage}{0.23\textwidth}
        \centering
        Fabian Kuhn\\
        \small University of Freiburg\\
        \small Freiburg, Germany\\
        \small \href{mailto:kuhn@cs.uni-freiburg.de}{\color{black}kuhn@cs.uni-freiburg.de}
    \end{minipage}
    \and
    \begin{minipage}{0.33\textwidth}
        \centering
        Zahra Parsaeian\\
        \small University of Freiburg\\
        \small Freiburg, Germany\\
        \small \href{mailto:zahra.parsaeian@cs.uni-freiburg.de}{\color{black}zahra.parsaeian@cs.uni-freiburg.de}
    \end{minipage}
}

\date{}

\begin{document}
\maketitle

\setcounter{page}{0}
\thispagestyle{empty}
\begin{abstract}
	Understanding the role of randomness when solving locally checkable labeling (LCL) problems in the LOCAL model has been one of the top priorities in the research on distributed graph algorithms in recent years. For LCL problems in bounded-degree graphs, it is known that randomness cannot help more than polynomially, except in one case: if the deterministic complexity of an LCL problem is in $\Omega(\log n)$ and its randomized complexity is in $o(\log n)$, then the randomized complexity is guaranteed to be $\poly(\log \log n)$. However, the fundamental question of \emph{which} problems with a deterministic complexity of $\Omega(\log n)$ can be solved exponentially faster using randomization still remains wide open.
	
	We make a step towards answering this question by studying a simple, but natural class of LCL problems: so-called degree splitting problems. These problems come in two varieties: coloring problems where the edges of a graph have to be colored with $2$ colors and orientation problems where each edge needs to be oriented. For $\Delta$-regular graphs (where $\Delta=O(1)$), we obtain the following results.
	
	\begin{itemize}
		\item We gave an exact characterization of the randomized complexity of all problems where the edges need to be colored with two colors, say red and blue, and which have a deterministic complexity of $O(\log n)$.
		\item For edge orientation problems, we give a partial characterization of the problems that have a randomized complexity of $\Omega(\log n)$ and the problems that have a randomized complexity of $\poly\log\log n$.
	\end{itemize}
	
	While our results are cleanest to state for the $\Delta$-regular case, all our algorithms naturally generalize to nodes of any degree $d<\Delta$ in general graphs of maximum degree $\Delta$.
\end{abstract}
\newpage

\section{Introduction and Related Work}
\label{sec:intro}

Understanding the complexity of locally checkable labeling (LCL)
problems is at the core of the area of distributed graph
algorithms. LCLs were introduced by Naor and Stockmeyer in
\cite{naor95} as labelings of bounded-degree graphs with labels from a
finite set of labels such that the correctness of a labeling can be
verified by checking a constant neighborhood of each node of the
graph. Classic examples for LCL problems are various coloring problems or the problem of computing a maximal independent set (MIS).

\subparagraph{Deterministic LCL Complexity Landscape.} Since the publication of \cite{naor95}, researchers have made astonishing progress in understanding the distributed complexity of solving LCL problems in the \LOCAL model\footnote{In the
	\LOCAL model~\cite{peleg00}, the nodes of a graph can exchange messages of
	arbitrary size over the edges of the graph in synchronous
	rounds. For a formal definition, we refer to
        \Cref{sec:model}.} and in related models. Apart from studying
      the complexity of specific problems (see, e.g.,
      \cite{BalliuHOS19,BrandtFHKLRSU16,Ghaffari2020ImprovedDD,Linial1992,KuhnMW16,GhaffariKuhn21,fraigniaud16local,Kuhn2009,LenzenW11,BarenboimE11}
      for a small selection), there also is substantial work on
      establishing the landscape of possible LCL complexities
      \cite{BalliuBOS18,BalliuHKLOS18,Balliu0OS20,chang16,ChangP17,FOCS18-derand,stoc17_complexity,Chang20,GrunauR022}
      and also towards understanding the task of deciding if a given
      LCL problem belongs to a certain complexity
      class~\cite{Balliu0CORS19,Chang20,BrandtHKLOPRSU17,binarylabelings,naor95,Balliu0COSS22}. For
      general bounded-degree graphs, we know that the deterministic
      \LOCAL model complexity of LCL problems is either $O(1)$, in the
      range $\Omega(\log\log^* n)$--$O(\log^* n)$, or in the range
      $\Omega(\log n)$--$O(n)$. By using techniques of \cite{naor95},
      Chang and Pettie~\cite{ChangP17} observed that there cannot be any LCL problem with a (deterministic or randomized) complexity in the range $\omega(1)$--$o(\log\log^* n)$. Further, Chang, Kopelowitz, and Pettie~\cite{chang16} proved that there cannot be any deterministic complexities between $\omega(\log^* n)$ and $o(\log n)$. It has further been shown in \cite{BalliuBOS18,BalliuHKLOS18} that the complexity landscape in the range $\Omega(\log n)$--$O(n)$ is dense, i.e., there are no further gaps. In bounded-degree trees, the landscape is even more sparse. In this case, the possible deterministic LCL complexities are $\Theta(1)$, $\Theta(\log^* n)$, $\Theta(\log n)$, and $\Theta(n^{1/k})$ for every integer $k\geq 1$~\cite{ChangP17,GrunauR022,Balliu0COSS22}. In bounded-degree trees, it is further known that for problems of complexity $\Omega(\log n)$, determining the exact asymptotic complexity of a given problem is decidable~\cite{Chang20}.

\subparagraph{Randomized LCL Complexities.} Understanding the role and power of randomness in the context of distributed LCL problems has been a major objective of the research in the area (see, e.g., \cite{stoc17_complexity,chang16,ChangP17,FOCS18-derand,Barenboim2013,Balliu0OS20,GhaffariK19} for papers that explicitly deal with this topic). As many LCL problems require some kind of local symmetry breaking, the use of randomness in distributed LCL algorithms is often quite natural. Recently, a series of papers~\cite{stoc17_complexity,FOCS18-derand,GGHIR23,GGR2020,GhaffariGrunauFOCS24,Rozhon2020} established generic derandomization techniques that allow to turn any $T$-round randomized \LOCAL algorithm for an LCL problem into a $T\cdot\tilde{O}(\log^2 n)$-round deterministic \LOCAL algorithm.\footnote{We use $\tilde{O}(x)$ to denote expressions that are upper bounded by $x\cdot \poly\log x$.} Using randomization can therefore not gain more than a factor of $\tilde{O}(\log^2 n)$. It is further known that for LCL problems of complexity $O(\log^* n)$, randomization does not help. There is however one important class of problems for which randomization provably helps exponentially. It was shown in \cite{BrandtFHKLRSU16,GhaffariS17,chang16} that computing a sinkless orientation of the edges of a bounded-degree graph requires $\Theta(\log n)$ rounds deterministically and $\Theta(\log\log n)$ rounds with randomization. Interestingly, this is a more general phenomenon. As shown in \cite{ChangP17}, on bounded-degree graphs, every randomized LCL algorithm with a round complexity of $o(\log n)$ can be sped up to run in the time required to solve an instance of the constructive Lov\'asz Local Lemma (LLL) problem with a polynomial LLL criterion (cf.\ \Cref{sec:model}). The randomized LLL round complexity is currently known to be $\tilde{O}(\log^4\log n)$. We therefore know that any randomized $o(\log n)$-round LCL algorithm can be sped up to run in $\tilde{O}(\log^4\log n)$ rounds. It is further known that any randomized $o(\log\log n)$-round LCL algorithm can be sped up to run in $O(\log^* n)$ rounds deterministically~\cite{chang16}. We therefore know that there is a class of problems that can be solved by using LLL algorithms that have polylogarithmic deterministic complexity and that have randomized complexity $\tilde{O}(\log^4\log n)$. In bounded-degree trees, it is even known that the complexity of the constructive LLL problem is exactly  $\Theta(\log\log n)$. Hence, in bounded-degree trees, some of the problems that can be solved in $\Theta(\log n)$ rounds deterministically can be solved in $\Theta(\log\log n)$ rounds with randomization and in the case of bounded-degree trees, these also are the only problems for which randomization helps. While we know that the randomized LLL complexity class exists, we do not yet have a thorough understanding of which problems belong to this family, not even for simple families of LCL problems in bounded-degree trees, raising the following important question.

\vspace{0.2cm}
\begin{mdframed}[innertopmargin=8pt,innerbottommargin=5pt,innermargin=5pt]
	\center
	\emph{Which} LCL problems (on bounded-degree trees or graphs) with a deterministic complexity of $\Omega(\log n)$ can be solved in sublogarithmic (and therefore in $O(\log \log n)$, resp.\ $O(\poly(\log \log n))$) rounds using randomization?
\end{mdframed}

It seems likely that providing a full characterization (i.e., designing an algorithm that decides for every LCL problem whether it can be solved with randomization in $o(\log n)$ rounds or requires $\Omega(\log n)$ rounds), if possible at all, will require a structured research program as such questions are notoriously hard.\footnote{For instance, it is not even known whether it is decidable whether a given LCL problem can be solved in constant time on bounded-degree trees.}
We make first progress towards this goal by providing such characterizations for some of the most fundamental problem classes amongst LCL problems exhibiting deterministic complexities of $\Omega(\log n)$ rounds.
Note that a variety of classic LCL problems, such as $(\Delta + 1)$-coloring, $(2\Delta - 1)$-edge coloring, maximal matching or maximal independent set, can be solved in $O(\log^* n)$ rounds on constant-degree graphs~\cite{goldberg1987parallel}, and therefore do not fall into the class of $\Omega(\log n)$-round problems.
Instead, the typical problems considered in the literature that fall in the deterministic $\Omega(\log n)$-regime are \emph{splitting problems} (which form the heart of classic divide-and-conquer approaches for solving other LCLs), which encompass both orientation and coloring problems.
In their standard form, splitting problems fall into the class of so-called \emph{binary labeling problems}---a class that has been instrumental in the context of complexity classification, including (efficient) decidability, in the deterministic setting~\cite{binarylabelings}.
Even restricted to binary labeling problems, the task of deciding between sublogarithmic and (at least) logarithmic randomized complexity is a major open problem on the path towards understanding randomness for LCL problems.

\subparagraph{Binary Labeling Problems.} In an attempt towards exactly characterizing the distributed complexities of all problems in a large class of LCL problems, the authors of \cite{binarylabelings} considered the class of binary labeling problems in the edge labeling formalism~\cite{Brandt19}. In this formalism, labels have to be assigned to the edges of a two-colored bipartite graph. For each side of the bipartite graph, the valid configurations are given by a set of allowed multisets of edge labels. The standard way to use this framework in general graphs is to consider the bipartite graph between nodes and edges. Each of the original nodes of a graph $G=(V,E)$ is colored white and in the middle of each edge $e\in E$, we place an additional black node. Each edge $e=\set{u,v}\in E$ is then split into two half-edges $(u,e)$ and $(v,e)$. The labels are placed on the half-edges and the validity of a labeling is determined by the allowed white node configurations and the allowed black edge configurations. The binary labeling problems are all problems that can be expressed with only two different labels (we call them red ($\red$) and blue ($\blue$) in the following) in the edge labeling formalism. In the graph setting with labels on half-edges, binary labeling problems correspond to edge orientation and edge $2$-coloring problems. If the allowed edge configurations are $\set{\red,\red}$ and $\set{\blue,\blue}$, we get problems where the edges of a graph need to be colored red and blue. If the only allowed edge configuration is $\set{\red,\blue}$, we get problems where all edges need to be oriented and if the allowed edge configurations are $\set{\red,\red}$ and $\set{\red,\blue}$, we get problems where we need to compute a partial orientation of the edges. All other combinations are symmetric or they lead to trivial cases. 

In \cite{binarylabelings}, the deterministic complexity of all binary labeling problems in the general edge labeling formalism in biregular bounded degree trees\footnote{More concretely, for two constants $d$ and $\delta$, it is assumed that the bipartite graph is properly colored with colors white and black and that all white nodes have degree $d$ and all black nodes have degree $\delta$. Note that when labeling half-edges in $\Delta$-regular general graphs, we have $d=\Delta$ and $\delta=2$.} was exactly characterized. We thus have an exact characterization of the deterministic complexity of all LCLs in regular bounded-degree trees that either require to color the edges with two colors or to (partially) orient the edges. In \cite{binarylabelings}, the authors also already prove some initial results on the randomized complexity of binary labeling problems.
However, as pointed out above, classifying the randomized complexity of those problems is still largely open and our results can be seen as a step towards closing this gap. We in particular give an exact characterization of the randomized complexity of solving arbitrary edge $2$-coloring problems in regular bounded-degree trees. %A detailed discussion of what we know and newly prove regarding the randomized complexity of binary labeling problems appears in \Cref{sec:contributions}.

\subparagraph{Degree Splitting Problems.} A special and important case of binary labelings are balanced oriented and unoriented degree splittings. Unoriented degree splittings are colorings of the edges with two colors such that every node $v$ is incident to approximately $\deg(v)/2$ edges of each color. An oriented degree splitting is an orientation of the edges so that for every node, indegree and outdegree are approximately equal. Degree splitting problems are not only natural to study, but also serve as key tools in other distributed tasks. Such problems lie at the heart of distributed divide-and-conquer algorithms and variants of splitting have in particular enabled efficient edge-coloring algorithms. In \cite{GhaffariS17, Ghaffari2020ImprovedDD}, splitting was used to obtain the first deterministic polylog-time edge coloring algorithm that uses only by a $(1 + o(1))$-factor more colors than the commonly considered $2\Delta - 1$ colors (thereby solving a long-standing open problem from the 90s). In \cite{GhaffariKMU18}, splitting is even used to obtain edge colorings with only $(1 + o(1))\Delta$ colors. Also, recent work on hypergraph sinkless orientation~\cite{BMNSU25} (which is a weak form of degree splitting) improved edge coloring and already inspired further follow-ups~\cite{JM25,JMS25}. Since our results provide stronger splitting guarantees, we expect them to similarly facilitate future advances, especially for randomized algorithms where previous progress has been limited. The distributed computation of degree splittings was in particular studied in \cite{GhaffariS17,Ghaffari2020ImprovedDD}.  In \cite{Ghaffari2020ImprovedDD} it was shown that one can compute both versions of degree splitting with constant discrepancy in time $\tilde{O}(\Delta)\cdot\log n$ deterministically and in time $\tilde{O}(\Delta)\cdot\log\log n$ randomized. The discrepancy is the maximum absolute difference between the two kinds of edges at any node. For oriented degree splitting, it was even shown that for odd-degree nodes, a discrepancy of $1$, and for even-degree nodes, a discrepancy of $2$ can be achieved. We improve the algorithm of \cite{Ghaffari2020ImprovedDD} and we show that also for unoriented degree splitting, the same discrepancy as for oriented degree splitting can be achieved.

\subsection{Our Contributions}\label{sec:contributions}
All our algorithms are based on existing subroutines for the following
three problems: \emph{sinkless orientation}, \emph{balanced
  orientation}, and \emph{constructive LLL}. The three problems have
different randomized and deterministic complexities and also different
complexities on general graphs and on trees.
As all our algorithms are deterministic reductions to three subroutines, it will be convenient to express all the algorithm complexities as functions of the complexities of the three subroutines. We therefore introduce the following notation. We use \TSO (for sinkless orientation), \TBO (for balanced orientation), and \TLLL (for constructive LLL) to denote the time complexities of the three subroutines in $n$-node graphs in the \LOCAL model. We formally define the three problems in \Cref{sec:model}. Below, we just give the concrete time bounds that the three shortcuts stand for. In the case of \TBO and \TLLL, the bounds depend on a degree parameter $d\geq 1$. For the randomized algorithms, the time bounds hold with high
probability. The time bounds for general graphs follow from
\cite{ChungPS14,FOCS18-derand,GGHIR23}, the randomized bound for
general graphs in addition requires the algorithm of
\cite{Davies23}. The time bounds for tree-like graphs follows from
\cite{ChangHLPU20}. Tree-like graphs in particular include graphs of
the form $T^{O(1)}$ if $T$ is a tree. Where not specified, the bounds
hold on general graphs.
\begin{eqnarray*}
	\TSO & = & O(\log n)\quad\text{for deterministic and}\quad O(\log\log n)\quad\text{for randomized algorithms}\\
	%\TSO & = & O(\log\log n)\quad\text{for randomized algorithms}\\
	\TBO(d) & = & d\cdot\log d\cdot \log^{1.71}\log d \cdot \TSO\\
	\TLLL(d) & = & \TSO\quad\text{on tree-like graphs}\\
	\TLLL(d) & = & \tilde{O}(\log^4 n)\quad\text{for deterministic algorithms on general graphs}\\
	\TLLL(d) & = & O\left(\frac{d}{\log d}\right) + \tilde{O}(\log^4 \log n)\quad\text{for randomized algorithms on general graphs}
\end{eqnarray*}

We will express the round complexities of all our algorithms as a deterministic function of the above complexities. To get the concrete bound for a specific case (trees or general graphs / deterministic or randomized algorithms), one just has to plug in the concrete bounds from above. For randomized algorithms, the resulting bound then also holds with high probability.

\subparagraph{Generalized Balanced Degree Splitting.}
As a first result, we generalize and also strengthen the main degree splitting result of \cite{Ghaffari2020ImprovedDD}. For this result, we use the edge labeling formalism for graphs (where we label half-edges), but in addition, to be able to express our generalization we also assume that each edge of a given graph $G=(V,E)$ has a binary input type. Concretely, we assume that each edge is either of type $\typeC$ (coloring edge) or of type $\typeO$ (orientation edge). A valid labeling of the half-edges with labels $\red$ and $\blue$ (red-blue coloring of the half-edges) must satisfy the following. For every edge $e$ of type $\typeC$, either both half-edges get assigned $\red$ or both half-edges get assigned $\blue$. For every edge $e$ of type $\typeO$, one half-edge must get label $\red$ and the other half-edge must get label $\blue$. For a given labeling of the half edges and a node $v\in V$, we use $\dred(v)$ and $\dblue(v)$ to denote the number of incident half-edges of $v$ that are labeled $\red$ and $\blue$ (colored red and blue), respectively.

\begin{theorem}\label{thm:balancedsplitting}
	Let $G=(V,E)$ be an $n$-node multigraph without selfloops, let $\Delta$ be the maximum degree of $G$, and assume that every edge $e\in E$ is either of type $\typeC$ or of type $\typeO$. Then, there is an $\TBO(\Delta)$-round algorithm to compute a red-blue coloring of the half-edges of $G$ such that for every node $v\in V$ of degree $d(v)$, it holds that
	\begin{equation}\label{eq:balancedsplitting1}
		\dred(v), \dblue(v) \in \left[\frac{d(v)}{2} - 1, \frac{d(v)}{2} + 1\right].
	\end{equation}
	If all edges are of type $\typeC$, we can even guarantee that
	\begin{equation}\label{eq:balancedsplitting2} 
		\dred(v) \in \set{\left\lfloor\frac{d(v)}{2}\right\rfloor, \left\lfloor\frac{d(v)}{2}\right\rfloor+1}
		\quad\text{and}\quad
		\dblue(v) \in \set{\left\lceil\frac{d(v)}{2}\right\rceil-1, \left\lceil\frac{d(v)}{2}\right\rceil}.
	\end{equation}
\end{theorem}
\begin{proof}
	Follows from \Cref{lemma:balancedsplitting1,lemma:balancedsplitting2} in \Cref{sec:balancedsplitting}.
\end{proof}

The above theorem strengthens and generalizes the result of \cite{Ghaffari2020ImprovedDD} in the following way. In \cite{Ghaffari2020ImprovedDD} it is shown that a balanced edge orientation where every node $v$ has indegree and outdegree in $\left[\frac{d(v)}{2} - 1, \frac{d(v)}{2} + 1\right]$ can be computed in time $\TBO(\Delta)$. That is, \Cref{eq:balancedsplitting1} of \Cref{thm:balancedsplitting} holds if all edges are of type $\typeO$. The authors of \cite{Ghaffari2020ImprovedDD} further prove that if all edges are of type $\typeC$, then one can achieve a partition into red and blue edges such that for all $v\in V$, we have $\dred(v), \dblue(v) \in \left[\frac{d(v)}{2} - 2, \frac{d(v)}{2} + 2\right]$. Implicitly, the techniques in \cite{Ghaffari2020ImprovedDD} also imply that the same result holds when the edges can be of type $\typeC$ or type $\typeO$.

\begin{corollary}\label{cor:balancedsplitting}
	Let $G=(V,E)$ be an $n$-node multigraph without selfloops, let $\eps>0$ be a parameter, and assume that every edge $e\in E$ is either of type $\typeC$ or of type $\typeO$. Then, there is an $\TBO(1/\eps)$-round algorithm to compute a red-blue coloring of the half-edges of $G$ such that for every node $v\in V$ of degree $d(v)$, $\dred(v), \dblue(v) \in \left[(1-\eps)\cdot \frac{d(v)}{2} - 1, (1+\eps)\cdot \frac{d(v)}{2} + 1\right]$.
\end{corollary}
\begin{proof}
	The corollary follows by splitting each node $v$ of $G$ into virtual nodes $v_1,\dots,v_k$ and partition the edges of $v$ among $v_1,\dots,v_k$ such that each virtual node gets between $\lceil 2/\eps\rceil$ and $2\cdot \lceil 2/\eps\rceil$ many edges  (or keep $v$ as it is if $d(v)<\lceil 2/\eps\rceil$). The number of virtual nodes per node is then at most $\lceil \eps d(v)/2\rceil$. We apply the algorithm of \Cref{thm:balancedsplitting} to the obtained virtual graph and use the obtained edge coloring on $G$. By  \Cref{eq:balancedsplitting1}, we then have
	\[
	\dred(v),\dblue(v) \leq \frac{d(v)}{2} + \left\lceil \eps \cdot \frac{d(v)}{2}\right\rceil \leq \left(1+\eps\right)\cdot \frac{d(v)}{2} + 1.\qedhere
	\]
%	as required.
\end{proof}

\subparagraph{Arbitrary Unoriented Red-Blue Splitting.}
As one of our main results, we generalize the second part of
\Cref{thm:balancedsplitting} (\Cref{eq:balancedsplitting2}) to
computing an arbitrary split into red and blue edges. We show that one
can efficiently compute a red-blue partitioning of the edges so that
the red degree (and consequently also the blue degree) of any node of
degree $\Delta$ only takes on one of the two values $y$ or $y+1$  (defined as the problem
$\Pi_{\Delta}(y)$ below). The red and blue
degrees of any node of degree $d<\Delta$ can take on one of three
different values, which depend on $d$ (defined as problem $\Pi(y)$ below). Concretely, we solve the following family of problems.

\begin{definition}[Arbitrary Red-Blue Coloring Problem]\label{def:arbitrary2splitting}
  Let $\Delta\geq 2$ and $y\in \set{0,\dots,\Delta-1}$ be two integer
  parameters. We define:
  \begin{itemize}
    \item{\rm\bf\boldmath Problem $\Pi_\Delta(y)$ for nodes of degree
    $\Delta$:} The problem $\Pi_{\Delta}(y)$ defines the set of
  $\red$/$\blue$-labelings of the edges $E$ of a graph $G=(V,E)$ of
  maximum degree $\Delta$ such that for all node $v\in V$ of degree
  $\Delta$, the number of incident red edges is $\dred(v)\in\set{y,
    y+1}$.
\item{\rm\bf\boldmath Problem $\Pi(y)$ for nodes of general degree:} 	
	The problem $\Pi(y)$ extends the definition and also requires non-trivial guarantees for nodes of degree $d<\Delta$. The problem $\Pi(y)$ defines the set of all $\red$/$\blue$-labelings of the edges such that a) $\Pi_\Delta(y)$ is solved and b) there exists vector $(x_1,\dots,x_{\Delta-1})$ such that for each $d\in \set{1,\dots,\Delta-1}$, we have
	\begin{itemize}
		\item Each node $v$ of degree $d$ has $\dred(v) \in \set{x_d - 1, x_d, x_d + 1}$.
		\item Let $\tau_d:=\frac{d}{2\Delta}$, $\tilde{x}_d:=\frac{d}{\Delta}\cdot\left(y+\frac{1}{2}\right)$, $\alpha_d:=\left\lfloor\tilde{x}_d\right\rfloor$, and $\beta_d:=\tilde{x}_d-\alpha_d$. If $\beta_d\leq \tau_d$, then $x_d=\alpha_d$, if $\beta_d\geq 1-\tau_d$, then $x_d=\alpha_d+1$, otherwise, $x_d\in\set{\alpha_d, \alpha_d+1}$.
                \end{itemize}
                \end{itemize}
\end{definition}

For nodes $v$ of degree $d$, the definition implies that the allowed red degree $\dred(v)$ is either $\lfloor (y+1/2)\cdot d/\Delta\rfloor \pm 1$ or $\lceil (y+1/2)\cdot d/\Delta\rceil \pm 1$ and if $(y+1/2)\cdot d/\Delta$ is close enough to an integer $x$, then $\dred(v)$ must be $x\pm 1$ for nodes of degree $d$. The next theorem shows that problem $\Pi(y)$ can be solved almost as efficiently as the balanced degree splittings of \Cref{thm:balancedsplitting}.

\begin{theorem}\label{thm:arbitrary2splitting}
	On $n$-node graphs with maximum degree $\Delta$, problem $\Pi(y)$ can be solved in $O(\log\Delta)\cdot \TBO(\Delta)$ rounds.
\end{theorem}

To prove \Cref{thm:arbitrary2splitting}, we assume that we are given a red/blue-partition of the edges of
$G$ such that every node $v$ of degree $\Delta$ has either $k$ or
$k+1$ red edges (and consequently either $\Delta-k-1$ or $\Delta-k$
blue edges). We can then use \Cref{eq:balancedsplitting2} of
\Cref{thm:balancedsplitting} to further split the red part (or the
blue part) in a balanced way and to recolor half of the old red part
blue (or half of the old blue part red). This can be done in such a
way that for nodes $v$ of degree $\Delta$, $\dred(v)$ and $\dblue(v)$
can still only take on two consecutive values. We will see that when
starting from a uniform coloring, by repeating this $O(\log\Delta)$
times (and possibly switching the colors in the end), one can solve
$\Pi(y)$ for every $y\in \set{0,\dots,\Delta-1}$.

We remark that when solving $\Pi(y)$ in $\Delta$-regular graphs, the
set of red edges forms what is known as a $\set{y,y+1}$-factor of the
graph $G$~\cite{graphfactors_book}. It was proven by
Tutte~\cite{tutte78} that every $\Delta$-regular graph contains a
$\set{y,y+1}$-factor for every $y\in\set{0,\dots,\Delta}$. An easy way
to see this is as a consequence of Vizing's theorem~\cite{vizing64} that every graph
of maximum degree $\Delta$ can be properly edge-colored with
$\Delta+1$ colors. Given $(\Delta+1)$-edge coloring of a
$\Delta$-regular graph, clearly all the edges of colors
$\set{1,\dots,y+1}$ form a $\set{y,y+1}$-factor (since every node can
miss at most one of those colors). The recent distributed
$(\Delta+1)$-edge coloring algorithm of Bernshteyn~\cite{Bernshteyn22} therefore
implies that a $\set{y,y+1}$-factor can be computed in time
$\mathrm{poly}(\Delta\cdot\log n)$ in the \LOCAL model. Prior to the
current work, this is the only somewhat efficient distributed
algorithm to compute a $\set{y,y+1}$-factor of regular graphs that we
are aware of. Note that \Cref{thm:arbitrary2splitting} implies that a
$\set{y,y+1}$-factor of a $\Delta$-regular graph can be computed in
randomized time $\tilde{O}(\Delta)\cdot\log\log n$. 

\subparagraph{Unbalanced Orientations.}
In \cite{Ghaffari2020ImprovedDD}, it was shown that one can compute a \emph{balanced orientation}, i.e., an orientation in which every node $v$ has indegree and outdegree approximately $\deg(v)/2$ in $O(\log\log n)$ randomized rounds in bounded-degree graphs. Further, in \cite{binarylabelings}, it was shown that some unbalanced orientation problems require $\Omega(\log n)$ rounds even with randomization in regular bounded-degree trees. Concretely, they show that this holds for a sufficiently large constant $c>0$ and orientations in which all nodes either have outdegree $\leq \Delta/4-c\sqrt{\Delta}$ or indegree $\leq \Delta/4 - c\sqrt{\Delta}$. We generalize both these results. The first theorem gives a family of unbalanced edge orientation problems that can be solved in $O(\log\log n)$ randomized rounds in trees and in $\poly\log\log n$ randomized rounds in general bounded-degree graphs.

\begin{theorem}\label{thm:orientationsUpper}
	Let $G=(V,E)$ be an $n$-node graph with maximum degree $\Delta$ and let $\rho_{1},\rho_{2}\in[0,1/2]$ be two parameters for which $\rho_1+\rho_2\geq 1/2$. There is a $\TLLL(\Delta^2)$-round algorithm to compute an edge orientation in which every node $v\in V$ either has outdegree at most $\rho_1\cdot\deg(v) + O(\sqrt{\Delta\log\Delta})$ or indegree at most $\rho_2\cdot\deg(v)+O(\sqrt{\Delta\log\Delta})$.
\end{theorem}
We present a random construction of the orientation of the edges that
is based on independent random variables on the nodes and edges. We
show that the probability that a node does not satisfy the condition
required by the theorem
%, i.e., every node $v \in V$ either has
%outdegree at most $\rho_1 \cdot \deg(v) + O(\sqrt{\Delta \log
%	\Delta})$ or $v$ has indegree at most $\rho_2 \cdot \deg(v) +
%      O(\sqrt{\Delta \log \Delta})$,
is polynomially small in $\Delta$, and
the orientation of each edge only depends on the random variables of
itself and its two endpoints. As a result, the problem of computing
such an orientation can be expressed as an LLL instance with
dependency degree $O(\Delta^2)$ and a polynomial LLL criterion. The
orientation can therefore be computed in time $\TLLL(\Delta^2)$ by
using existing distributed LLL algorithms.
We also generalize the lower bound that was proven in \cite{binarylabelings}.

\begin{theorem}\label{thm:orientationsLower}
	Let $\Delta\geq 1$ be an integer, let $\rho_{1},\rho_{2}\geq 0$ be two parameters for which $\rho_{1}+\rho_{2} = 1/2$, and let $c>0$ be a sufficiently large constant. In $\Delta$-regular trees, computing an orientation for which every node $v$ either has outdegree at most $\rho_{1}\cdot\Delta - c\sqrt{\Delta}$ or indegree at most $\rho_{2}\cdot\Delta - c\sqrt{\Delta}$ requires $\Omega(\log_\Delta n)$ rounds in the \LOCAL model.
\end{theorem}

To prove the theorem, we show how to construct a $\Delta$-regular graph with girth $\Omega(\log_{\Delta} n)$ for which an orientation for which every node has outdegree at most $\rho_{1}\cdot\Delta - c\sqrt{\Delta}$ or indegree at most $\rho_{2}\cdot\Delta - c\sqrt{\Delta}$ (where $\rho_{1}+\rho_{2} = 1/2$) does not exist if the constant $c$ is chosen sufficiently large. As in $o(\log_{\Delta} n)$ rounds, nodes cannot locally distinguish a $\Delta$-regular tree from a high-girth $\Delta$-regular graph, this implies that in $\Delta$-regular trees, we need $\Omega(\log_{\Delta} n)$ rounds to compute the desired orientation.

\Cref{thm:orientationsUpper,thm:orientationsLower} approximately
capture the edge orientations that can be computed in $O(\log\log n)$
randomized rounds in $\Delta$-regular trees or in $\poly\log\log n$
rounds in $\Delta$-regular graphs in the following way. Any locally
checkable orientation problem (without input) in $\Delta$-regular
graphs can be expressed as a list of allowed outdegrees. Let $d_1,
d_2, \dots, d_k$ be the outdegrees that a node of degree $\Delta$ is
allowed to have in a valid solution to a given problem. We can define
the $\theta$-smoothed version of this problem as the problem that
allows all outdegrees $d$ for which there is a value $d_i$ ($i\in
\set{1,\dots,k}$) for which $d\in [d_i-\theta,d_i+\theta]$. On the one
hand, \Cref{thm:orientationsUpper} shows that whenever there exists
$i,j\in \set{1,\dots,k}$ such that $d_i \leq
\Delta/2+O(\sqrt{\Delta\log\Delta})$, $d_j\geq\Delta/2 -
O(\sqrt{\Delta\log\Delta})$, and $d_j-d_i\leq
\Delta/2+O(\sqrt{\Delta\log\Delta})$ (we can have $i=j$), then an
$O(\sqrt{\Delta\log\Delta})$-smoothed version of the problem can be
solved in time $\poly\log\log n$. On the other hand,
\Cref{thm:orientationsLower} shows that whenever there does not exist
two such values, then for a small enough constant $c>0$, not even a
$c\sqrt{\Delta}$-smoothed version of the problem can be
solved in time $\poly\log\log n$.

\smallskip

The technical details of
\Cref{thm:orientationsUpper,thm:orientationsLower} appear in
\Cref{sec:unbalanced}. Further, in \Cref{sec:smallmsg}, we discuss to what extent our algorithms can be implemented with small messages.

\section{Model and Definitions}
\label{sec:model}

\subparagraph{Mathematical Definitions and Preliminaries.} For a (multi)graph $G=(V,E)$ and a node $v\in V$, we use $d_G(v)$ to denote the degree of $v$ (in the case of a multigraph, the number of incident edges of $v$). Whenever $G$ is clear from the context, we write $d(v)$ instead of $d_G(v)$. Further, whenever $G$ denotes the communication graph and $G$ is clear from the context, we use $n$ to denote the number of nodes of $G$ and $\Delta$ to denote the maximum degree of $G$.

\subparagraph{Communication Model.} In the \LOCAL
model~\cite{Linial1992,peleg00}, the nodes of an $n$-node graph
$G=(V,E)$ communicate in synchronous rounds over the edges of
$G$. The nodes are equipped with unique $O(\log n)$-bit IDs and in each
round, each node $v\in V$ can perform some arbitrary internal computation,
send a message of arbitrary size to each neighbor, and receive
the messages sent to it by its neighbors. The time (or round)
complexity of an algorithm is defined as the maximum number of rounds
required for all nodes to terminate. The time complexity is typically
expressed as a function of global network parameters such as $n$ or
the maximum degree $\Delta$. We
assume that those global parameters are known to the nodes of
$G$.\footnote{This assumption can usually be dropped or at least
  significantly relaxed %(possibly at the cost of a more technical
  % algorithm/analysis)
  by only requiring the knowledge of some approximations of the parameters~(see, e.g., \cite{KormanSV13}).}

\subparagraph{Basic Subroutines.} As described in \Cref{sec:contributions}, our algorithms use the following three subroutines in our algorithms.
\begin{itemize}
	\item \emph{Sinkless Orientation:} The sinkless orientation problem on an $n$-node multigraph $G=(V,E)$ asks for an orientation of the edges $E$ such that every node $v\in V$ of degree $\geq 3$ has at least one outgoing edge. It was shown in \cite{GhaffariS17,Ghaffari2020ImprovedDD} that the sinkless orientation can be solved in $O(\log n)$ rounds deterministically and in $O(\log\log n)$ rounds with randomization.
	\item \emph{Balanced Orientation:} The problem of computing a balanced orientation of an $n$-node multigraph $G=(V,E)$ asks for an orientation of the edges $E$ such that every node $v\in V$ has outdegree $\in [d(v)/2-1, d(v)/2 + 1]$. If $\Delta$ is the maximum degree of $G$, it is shown in \cite{Ghaffari2020ImprovedDD} that a balanced orientation can be computed in time $O(\Delta\cdot \log\Delta\cdot\log^{1.71}\Delta)\cdot \TSO$.
	\item \emph{Constructive LLL:} The Lov\'asz Local Lemma (LLL) considers the following setting. A set $V_H$ of (bad) events that are defined as functions over some underlying product probability space. Let $H=(V_H,E_H)$ be the \emph{conflict graph} of those events, i.e., $H$ contains an edge between two events iff they are dependent (i.e., if they depend on intersecting sets of random variables in the underlying product space). The LLL states that if the maximum degree of $H$ is $d$ and the probability of each event is upper bounded by $p$ such that $ep(d+1)\leq 1$, then the probability that none of the events occurs is positive. The instance is said to have a polynomial LLL criterion if $p\leq d^{-\gamma}$ for a constant $\gamma>0$. If $\gamma$ is a sufficiently large constant, there are efficient distributed algorithms to determine a setting for the underlying random variables that avoids all the bad events. The best known round complexities for those algorithms are $O(\log n)$ for deterministic algorithms in tree-like dependency graphs~\cite{ChangHLPU20}, $O(\log\log n)$ for randomized algorithms in tree-like dependency graphs~\cite{ChangHLPU20}, $\tilde{O}\big(\log^4 n\big)$ for deterministic algorithms in general graphs~\cite{ChungPS17,FOCS18-derand,GGHIR23}, and $O\big(\frac{d}{\log d}\big) + \tilde{O}(\log^4\log n)$ for randomized algorithms in general graphs~\cite{ChungPS17,FOCS18-derand,GGHIR23,Davies23}.
\end{itemize}

\section{Generalized Balanced Degree Splitting}
\label{sec:balancedsplitting}

The objective of this section is to prove \Cref{thm:balancedsplitting}. We do this in several steps. We first prove that in asymptotically the same time as computing a sinkless orientation, one can compute a red-blue coloring of the half-edges of a graph with edge types such that each node sees at least one half-edge of each color.

\begin{lemma}\label{lem:reduce_to_SO}
	Let $G=(V,E)$ be an $n$-node multigraph without selfloops and assume that every edge $e\in E$ is either of type $\typeC$ or of type $\typeO$. Then, there is an $\TSO$-round distributed algorithm to compute a coloring of the half-edges of $G$ that is consistent with the edge types and for which every node $v\in V$ of degree $d(v)\geq 3$ has $\dred(v)\geq 1$ and $\dblue(v)\geq 1$.
\end{lemma}

\begin{proof}
	In the following, we describe each step of the algorithm separately.

	\subparagraph*{Reducing to Maximum Degree 3.} We first select
        a subgraph $G'=(V',E')$ of $G$ on which we solve the
        problem. The node set of $G'$ is $V'=V$ and to construct the
        edges $E'$, each $v \in V$ with $d_G(v) \geq 3$ arbitrarily marks $3$ of its edges. If an edge $e = \{u, v\}$ is marked by both $u$ and $v$, it is included in $E'$. Since an edge is added to $E'$ only if both of its endpoints have marked it and each node marks at most $3$ edges, the maximum degree of $G'$ is $\leq 3$. 
	
	Given a valid red-blue coloring of the half-edges of $G'$ \footnote{A valid red-blue coloring is a red-blue assignment on the half-edges of a graph that aligns with the edge types and satisfies the vertex constraints of the graph.}, we can construct a valid red-blue coloring on $G$ as follows. Nodes of degree $3$ in $G'$ already have at least one incident half-edge of each color. All other nodes $v$ of degree at least $3$ in $G$ have $3-d_{G'}(v)$ marked edges that are not in $G'$. Those edges are not marked by their other node and $v$ can therefore color them arbitrarily. If there are at least two such edges, $v$ colors them such that the two corresponding incident half-edges get different colors. If it has only one such edge, it already has incident half-edges of at least one color and it can color the remaining marked edge with the other color. The conversion from $G$ to $G'$ and computing a solution on $G$ when given a solution on $G'$ both clearly require only $O(1)$ rounds. Given a valid red-blue for $G'$, we can derive a valid red-blue coloring for the original graph $G$. Therefore, in the following, we describe how to create a valid red-blue coloring for $G'$.
	
	\subparagraph{Clustering the Graph.} Next, we compute a $(3,2)$-ruling set $S$ on $G'$. That is, we select a set of nodes $S$ such that for any two nodes $u,v\in S$, the distance between $u$ and $v$ in $G'$ is at least $3$ and for any node $u\not\in S$, there is a node in $S$ within distance $2$ in $G'$. Such a set $S$ can be computed in $O(\log^* n)$ rounds by running an MIS algorithm on $G'^2$~\cite{Kuhn2009}. Note that $G'^2$ is a bounded-degree graph and the \LOCAL model on $G'^2$ can be simulated with constant overhead in the \LOCAL model on $G$. Given $S$, we partition $G'$ into clusters. We build one cluster for each node $v\in S$, where each node $u\not\in S$ joins the cluster of the nearest $S$-node (ties broken arbitrarily).
	Note that since nodes in $S$ have distance at least $3$ from each other, each cluster consists at least its cluster center $v\in S$ and all neighbors $u$ of $v$ in $G'$. Further, since every node is at distance at most $2$ from a cluster center, the cluster of $v\in S$ is contained within $v$'s $2$-hop neighborhood.
	
	Clusters are called \emph{good} if they include a vertex with degree less than $3$ or if they contain a cycle, and \emph{bad} otherwise. Note that the subgraph induced by the nodes of a bad cluster is a tree of height at most $2$ (when rooted at the cluster center $v\in S$). Since any bad cluster centered at $v\in S$ contains at least the three neighbors of $v$, each bad cluster has at least six intercluster edges, where an intercluster edge is an edge that connects endpoints in two distinct clusters. In the following, we first color the internal edges of the bad clusters and the intercluster edges, and we afterwards color the internal edges of good clusters.

	\subparagraph{Assigning Intercluster Edges to Clusters.}
        Before coloring any edges, we assign each intercluster edge to
        one of its two incident clusters. To this end, we construct a multigraph $H=(V_H,E_H)$ as follows. The graph $H$ contains one node $v_C\in V_H$ for each good cluster $C$ of $G'$ and two nodes $v_{C,1}$ and $v_{C,2}$ for each bad cluster $C$ of $G'$. Further, the edge set $E_H$ contains one edge for each intercluster edge of $G'$. An intercluster edge $\set{u,v}\in E'$ is assigned to an edge in $E_H$ as follows. The edge is between two cluster nodes in $V_H$ corresponding to the clusters containing $u$ and $v$. For each bad cluster $C$, we make sure that both nodes $v_{C,1}$ and $v_{C,2}$ obtain at least $3$ edges. Note that this is possible because each bad cluster is incident to at least $6$ intercluster edges.  We then apply the sinkless orientation algorithm on this modified graph, which requires $\TSO$ time. As a result, any node of $H$ of degree at least $3$ is ensured to have at least one outgoing edge. We assign each edge in $E_H$ to the node in $V_H$ for which it is outgoing. This directly induces an assignment of the intercluster edges of $G'$ to one of the adjacent clusters. Note that because each bad cluster has two  nodes of degree $\geq 3$ in $H$, each bad cluster is assigned at least two of its intercluster edges.
	
	\subparagraph{Handling the Bad Clusters.} We can now assign the colors to all edges inside the bad clusters and to all intercluster edges. Recall that each bad cluster has at least two adjacent intercluster edges that are assigned to it. We arbitrarily name one of those edges $e_c$ and another one of those edges $e_{nc}$. We color the half-edges of all the intercluster edges that are assigned to a good cluster and the half-edges of all the intercluster edges that are assigned to a bad cluster and that are not selected as the edge $e_{nc}$ of this cluster. This coloring is done so that it is consistent with the types of the colored edges, but otherwise it is done arbitrarily. In each bad cluster, we now identify the (unique) path between $e_{nc}$ and $e_c$. We assume that the path includes $e_{nc}$, but it excludes $e_c$, \Cref{fig:bad-cluster} (a). We now $2$-color the half-edges on this path, starting from the edge incident to the endpoint of $e_c$ and moving towards $e_{nc}$, in accordance with the edge types in such a way that each inner node of the path has one adjacent red and one adjacent blue half-edge on the path. Because the cluster diameter is bounded, the path length is constant and thus, the $2$-coloring of the path can be done in $O(1)$ rounds. Assume that all the remaining edges inside the cluster are oriented towards the $e_c\text{-}e_{nc}$-path,  \Cref{fig:bad-cluster} (b). Those edges are now colored by going from the leaves of the internal cluster tree towards the $e_c\text{-}e_{nc}$-path. Consider such a remaining edge $\set{u,v}$ inside the cluster and assume that $v$ is the parent of $u$ (i.e., $v$ is the node closer to the $e_c\text{-}e_{nc}$-path). When coloring the edge $\set{u,v}$, node $u$ already has two colored half-edges. We can therefore color the edge $\set{u,v}$ so that node $v$ is happy, i.e., has two incident half-edges of different colors,  \Cref{fig:bad-cluster} (c). In this way, we can make all nodes happy except the nodes on the path. However, we have already made sure that the path nodes see both colors on their two edges on the path. In this way, we can therefore make all the nodes in the bad clusters happy.
	
		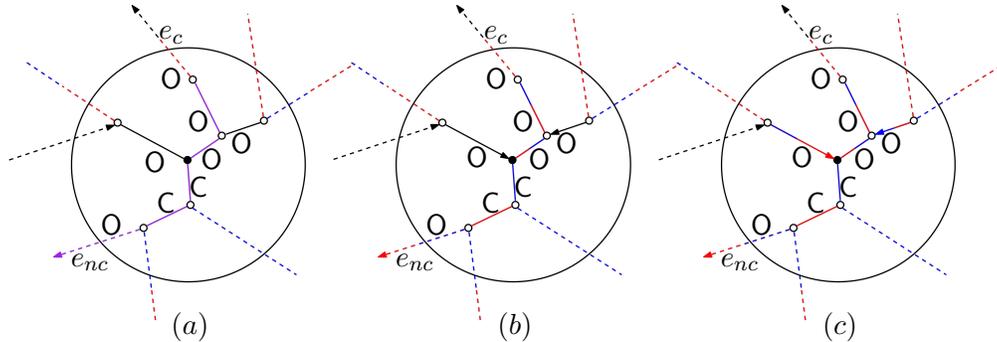
\begin{figure}
		\centering
		\begin{tikzpicture}[scale=0.72]
			\node at (0,0){\includegraphics[scale=.1]{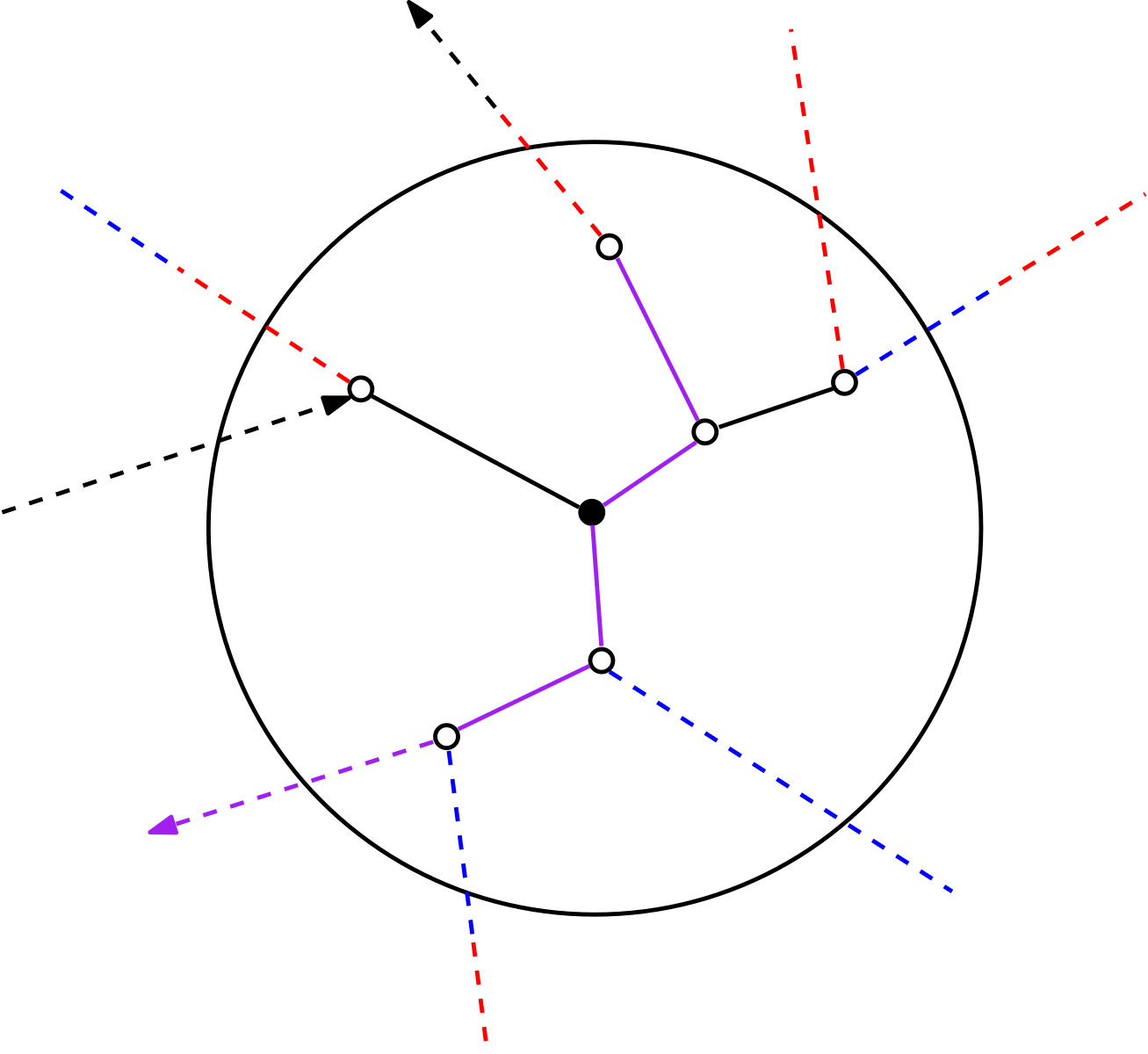}};
			\node at (-.5,.05) {$\typeO$};
			\node at (.55,.11) {$\typeO$};
			\node at (.3,-.4) {$\typeC$};
			\node at (-.3,-.7) {$\typeC$};
			\node at (-1.3,-1.1) {$\typeO$};
			\node at (-1.7,-1.8) {$e_{nc}$};
			\node at (1.1,.4) {$\typeO$};
			\node at (.3,.8) {$\typeO$};
			\node at (-.2,1.6) {$\typeO$};
			\node at (-.2,2.4) {$e_c$};
			\node at (0.15, -3) {$(a)$};
			\node at (6,0){\includegraphics[scale=.1]{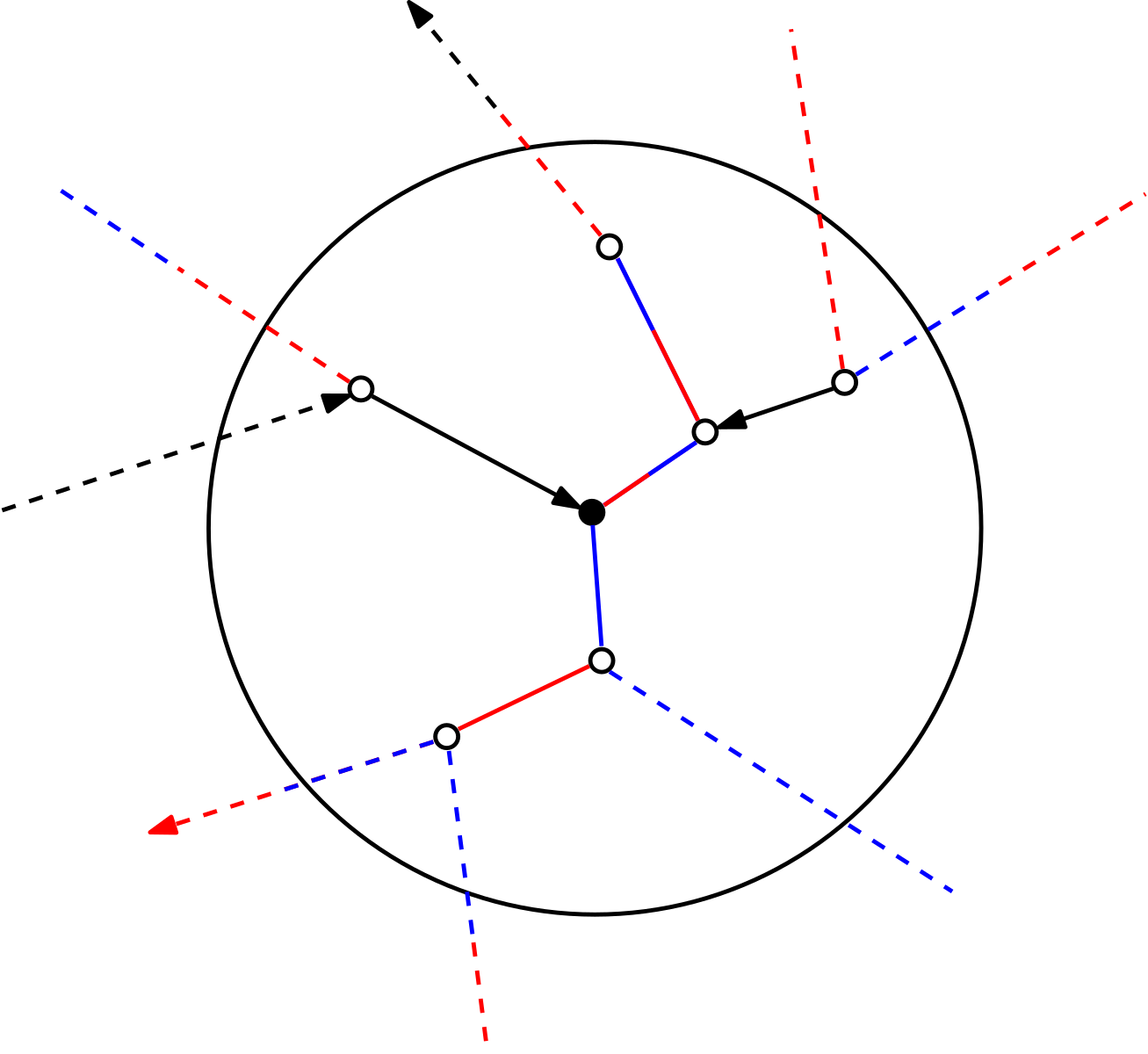}};
			\node at (5.45,.15) {$\typeO$};
			\node at (6.6,.11) {$\typeO$};
			\node at (6.3,-.4) {$\typeC$};
			\node at (5.7,-.7) {$\typeC$};
			\node at (4.7,-1.1) {$\typeO$};
			\node at (4.3,-1.8) {$e_{nc}$};
			\node at (7.1,.4) {$\typeO$};
			\node at (6.3,.8) {$\typeO$};
			\node at (5.8,1.6) {$\typeO$};
			\node at (5.8,2.4) {$e_c$};
			\node at (6.15, -3) {$(b)$};
			\node at (12,0){\includegraphics[scale=.1]{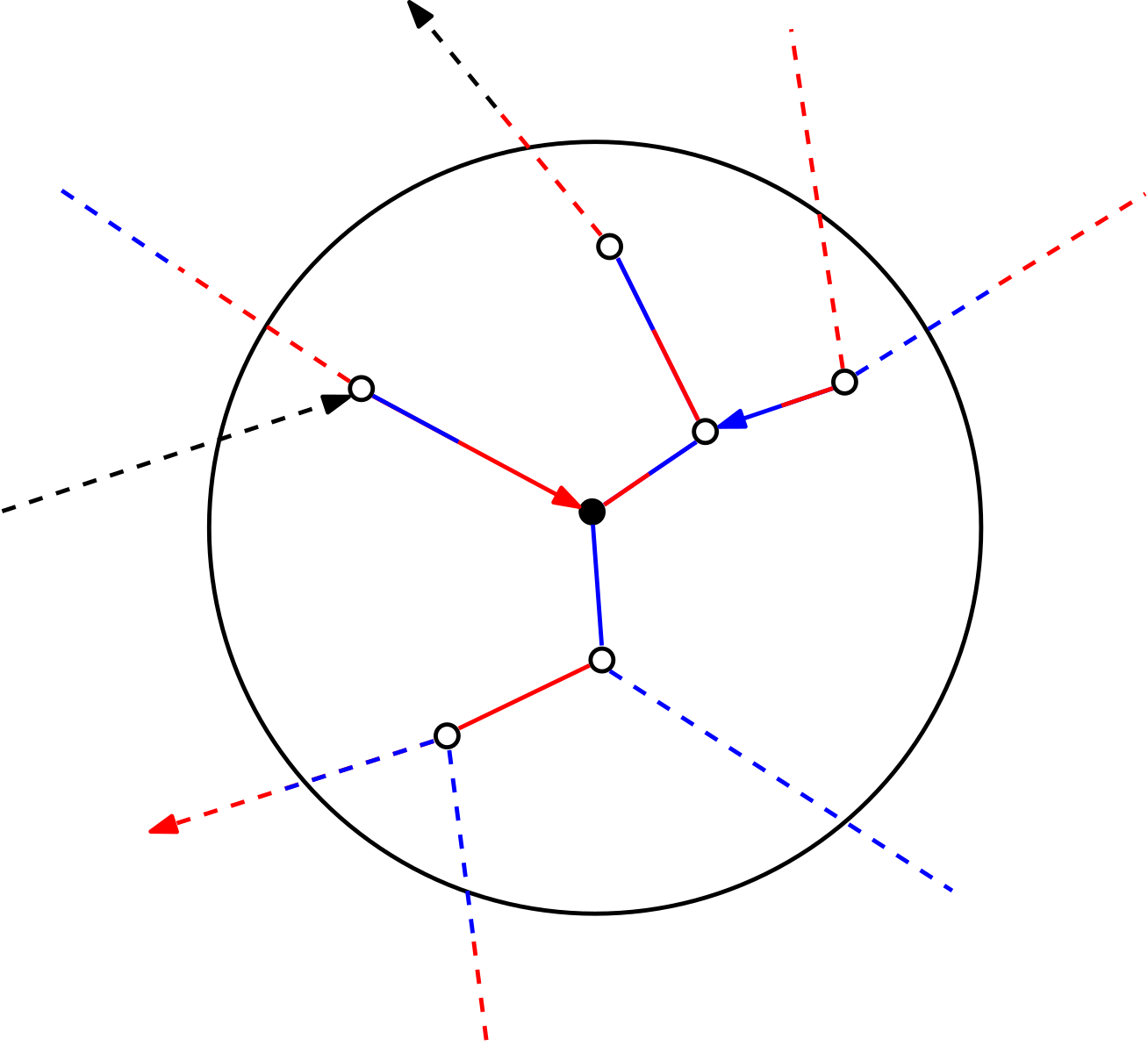}};
			\node at (11.45,.15) {$\typeO$};
			\node at (12.6,.11) {$\typeO$};
			\node at (12.3,-.4) {$\typeC$};
			\node at (11.7,-.7) {$\typeC$};
			\node at (10.7,-1.1) {$\typeO$};
			\node at (10.3,-1.8) {$e_{nc}$};
			\node at (13.1,.4) {$\typeO$};
			\node at (12.3,.8) {$\typeO$};
			\node at (11.8,1.6) {$\typeO$};
			\node at (11.8,2.4) {$e_c$};
			\node at (12.15, -3) {$(c)$};
		\end{tikzpicture}
		\caption{Coloring the edges of a bad cluster}
		\label{fig:bad-cluster}
	\end{figure}

	\subparagraph{Handling the Good Clusters.} We can move on to the good clusters. Note that at this point, all the intercluster edges are colored and it therefore only remains to color the internal edges of all the good clusters. Let us first assume that we have a good cluster that contains at least one node of degree at most $2$ and let $v_2$ be such a node. Node $v_2$ has a degree less than $3$ and nodes of degree less than $3$ do not have to satisfy any constraints. Therefore, all the edges incident to $v_2$ can be colored freely (from $v_2$'s point of view). We now construct a spanning tree $T$ of the internal edges of the cluster and we orient the edges of $T$ towards $v_2$. We arbitrarily color all edges in the cluster that do not belong to $T$. We then proceed to $2$-color the half-edges of $T$ starting from the leaves, moving towards $v_2$. Going up the oriented tree $T$, we can make each node happy except for the root of $T$, which is happy by default because its degree is less than $3$. This process is also completed within $O(1)$ rounds.
	
	We next turn our attention to good clusters that contain only nodes of degree $3$, but include a cycle. If a cluster contains only one cycle, we select it. If there are multiple cycles, we select an arbitrary one, \Cref{fig:good-cluster} (a). Subsequently, we contract the selected cycle into a single vertex $v$ and we construct a spanning tree $T$ of the remaining graph in the cluster. We arbitrarily color all edges in the cluster that do not belong to $T$ or the selected cycle. The tree edges are colored in the same way as before, \Cref{fig:good-cluster} (b). However, this time, the root node corresponds to the nodes of a cycle in $G'$ and thus, the root is not automatically happy. However, since each node on the cycle has degree $3$ and exactly two incident edges within the cycle, after coloring all the edges except for the cycle edges, each node of the cycle already has one colored incident edge. To color the edges of the cycle, we consistently orient the cycle in one of the two directions. Now, each node $v$ in the cycle is responsible for coloring its outgoing edge and it colors it to make itself happy, \Cref{fig:good-cluster} (c). This concludes the algorithm. Note that also the good clusters can clearly be processed in $O(1)$ rounds.   
	
	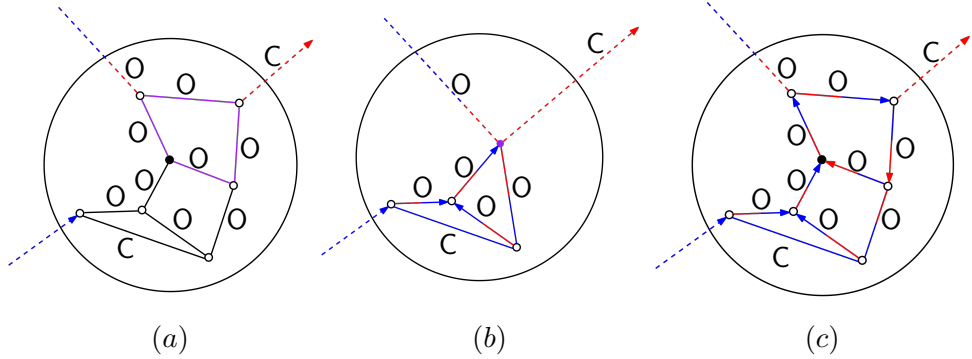
\begin{figure}
		\centering
		\begin{tikzpicture}[scale=0.85]
			\node at (0,0){\includegraphics[scale=.1]{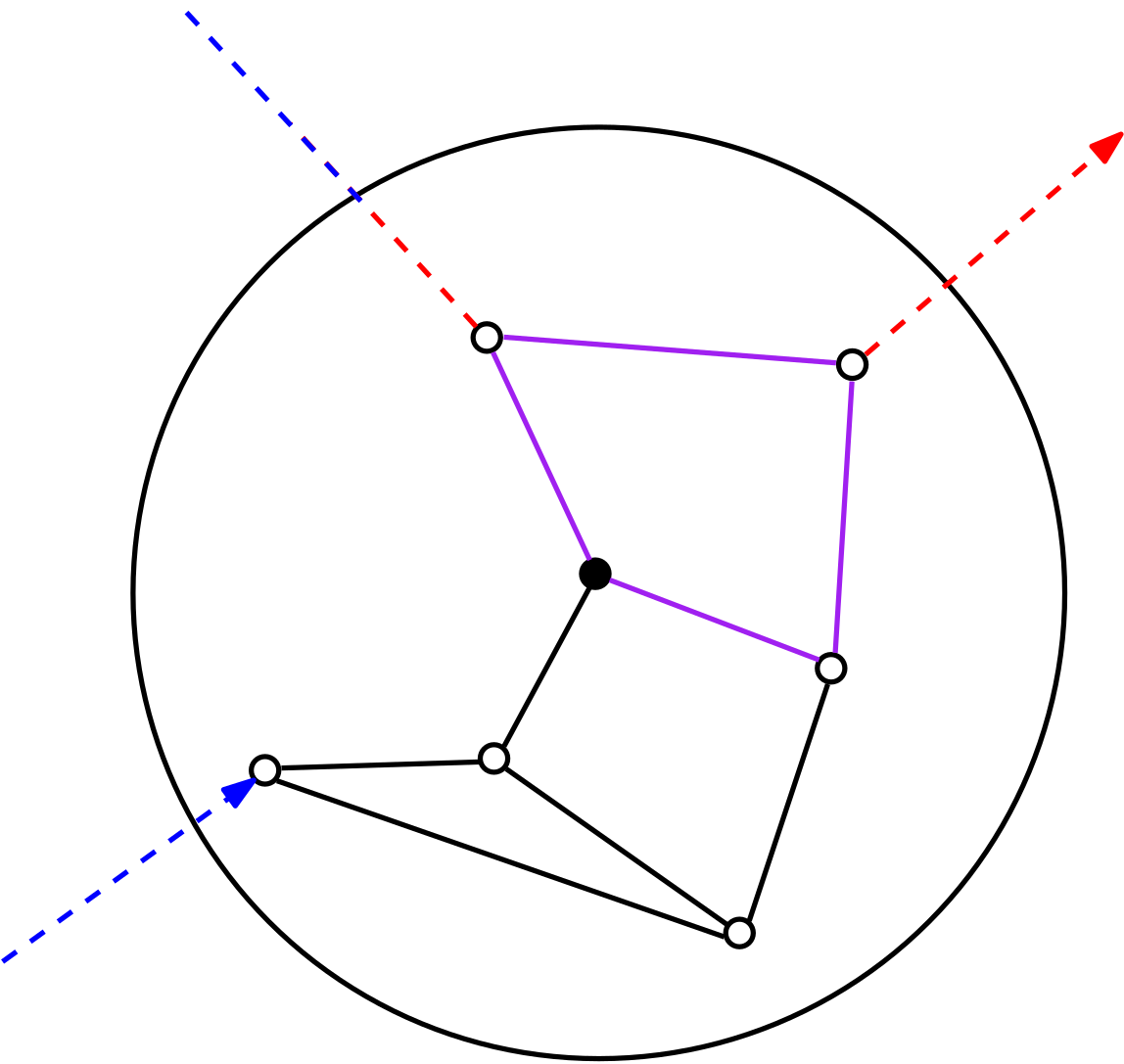}};
			\node at (5,0){\includegraphics[scale=.1]{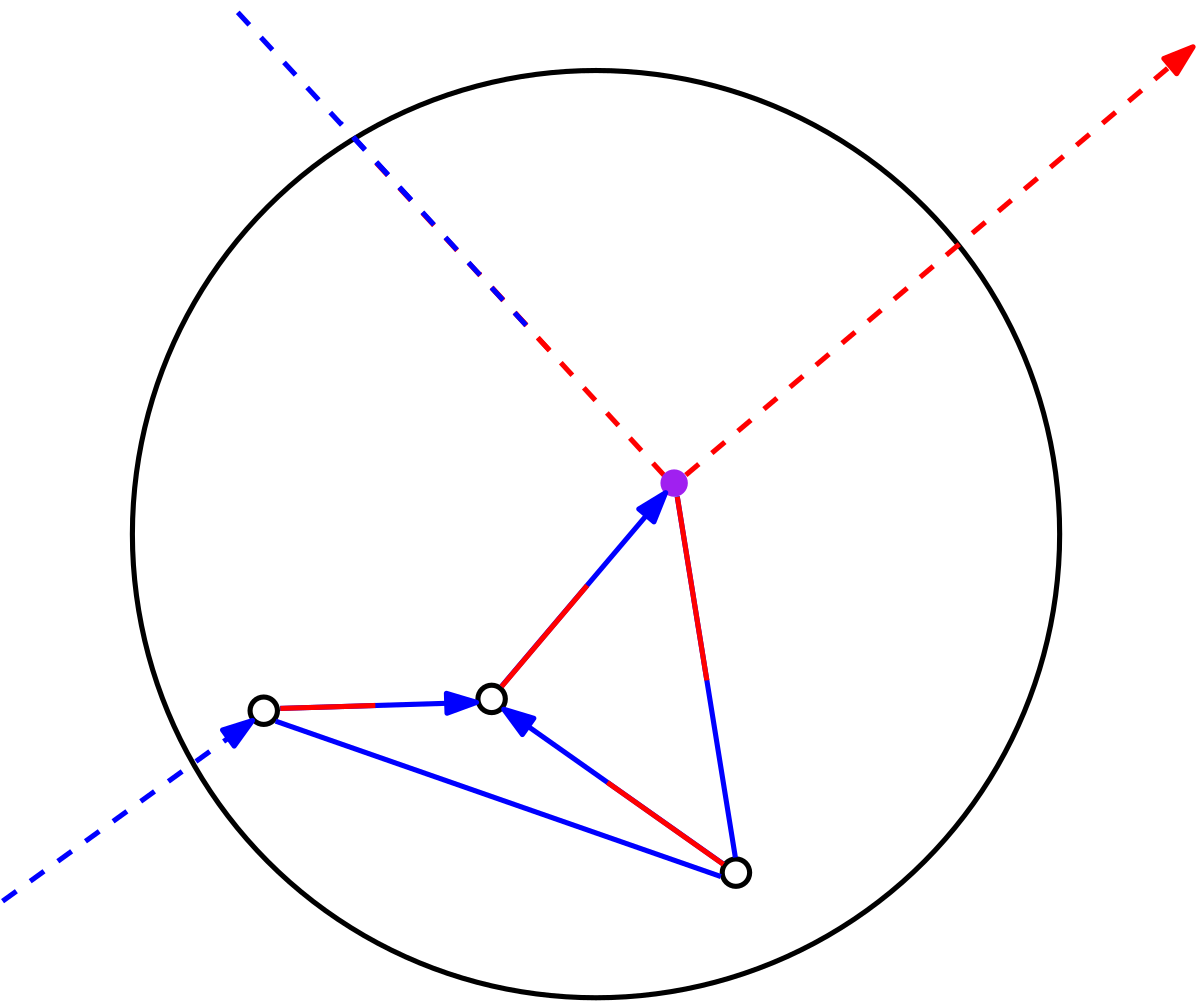}};
			\node at (10.2,0){\includegraphics[scale=.097]{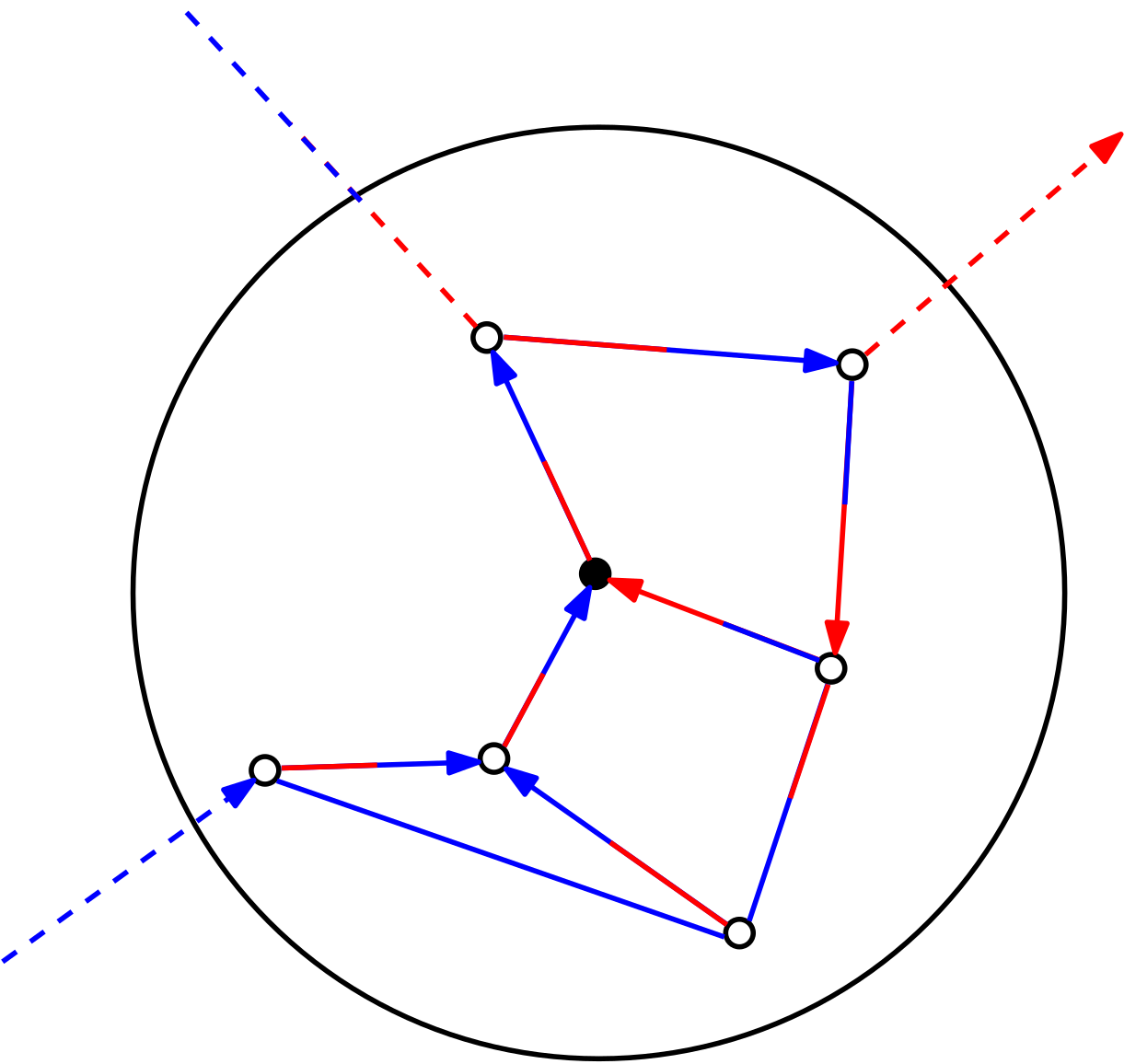}};
			\node at (-.4,1.25) {$\typeO$};
			\node at (1.75,1.45) {$\typeC$};
			\node at (-.35,.25) {$\typeO$};
			\node at (.4,1) {$\typeO$};
			\node at (1.4,0) {$\typeO$};
			\node at (.6,-.15) {$\typeO$};
			\node at (1.2,-1.1) {$\typeO$};
			\node at (.35,-1.1) {$\typeO$};
			\node at (-.25,-.5) {$\typeO$};
			\node at (-.7,-.75) {$\typeO$};
			\node at (-.55,-1.56) {$\typeC$};
			\node at (0.15,-3) {$(a)$};
			% \node at (-.3,2.4) {$e_c$};
			% \node at (6,0){\includegraphics[scale=.17]{./images/bad-cluster/12.png}};
			\node at (5.1,-.9) {$\typeO$};
			\node at (5.65,-.6) {$\typeO$};
			\node at (4.75,-.3) {$\typeO$};
			\node at (4.1,-.6) {$\typeO$};
			\node at (4.5,-1.5) {$\typeC$};
			\node at (6.8,1.65) {$\typeC$};
			\node at (4.7,1) {$\typeO$};
			\node at (5.15,-3) {$(b)$};
			% \node at (12,0){\includegraphics[scale=.17]{./images/bad-cluster/13.png}};
			\node at (9.8,1.25) {$\typeO$};
			\node at (12.05,1.5) {$\typeC$};
			\node at (9.95,.2) {$\typeO$};
			\node at (10.6,1.05) {$\typeO$};
			\node at (11.65,0) {$\typeO$};
			\node at (10.9,-.2) {$\typeO$};
			\node at (11.43,-1.1) {$\typeO$};
			\node at (10.45,-1.1) {$\typeO$};
			\node at (9.95,-.5) {$\typeO$};
			\node at (9.4,-.8) {$\typeO$};
			\node at (9.7,-1.7) {$\typeC$};
			\node at (10.35,-3) {$(c)$};
			% \node at (11.7,2.4) {$e_c$};
		\end{tikzpicture}
		\caption{Coloring the edges of a good cluster with cycle(s)}
		\label{fig:good-cluster}
	\end{figure}

	The most time-consuming part of the whole algorithm is the computation of the sinkless orientation of the cluster graph $H$. All other steps require either $O(1)$ or $O(\log^* n)$ rounds. The overall round complexity is therefore $\TSO$.
\end{proof}

We can now prove the first part of \Cref{thm:balancedsplitting}

\begin{lemma}\label{lemma:balancedsplitting1}
	Let $G=(V,E)$ be an $n$-node multigraph without selfloops and assume that every edge $e\in E$ is either of type $\typeC$ or of type $\typeO$. There is an $\TBO(\Delta)$-round \LOCAL algorithm for computing a red-blue coloring of the half-edges of $G$ such that for every node $v\in V$ of degree $d(v)$, it holds that
	\begin{equation}\label{eq:balancedsplitting3}
		\dred(v), \dblue(v) \in \left[\frac{d(v)}{2} - 1, \frac{d(v)}{2} + 1\right].
	\end{equation}
\end{lemma}
\begin{proof}
	The high-level idea of the proof is the following. For
	$k=\lceil\log_2\Delta\rceil$, we compute a sequence of multigraphs
	$G=G_0, G_1, G_2, \dots, G_{k+2}$ of maximum degree
	$\Delta=\Delta_0,\Delta_1,\dots,\Delta_{k+2}$ such that the maximum
	degree $\Delta_{i+1}$ of $G_{i+1}$ is roughly
	$\Delta_{i+1}\approx\Delta_i/2$ and such that $\Delta_{k+2}\leq
	4$. The graphs are constructed such that solving the problem
	required by \Cref{eq:balancedsplitting3} on any graph $G_i$ solves
	the same problem on $G$.
	
	We first show how to construct the graphs $G_0,\dots,G_k$ and we show how to construct $G_{k+1}$ and $G_{k+2}$ at the end of the proof.	For each $i\in\set{0,\dots,k-1}$, graph $G_{i+1}$ is constructed from graph $G_{i}=(V_i,E_i)$ as follows. For every node $v\in V_i$, we use $d_i(v)$ to denote the degree of $v$ in $G_i$. For each $i \in\set{0,\dots,k-1}$, we fix $\eps_i := 2^{i/2}/(4\sqrt{\Delta})$. Note that since $k=\lceil\log_2\Delta\rceil$, $\eps_i\leq 1/4$ for all $i\in\set{1,\dots,k-1}$.
	
	We now use Theorem 1.1 of \cite{Ghaffari2020ImprovedDD} to compute an edge orientation of $G_{i}$ in which every node $v$ has outdegree at least $(1 - \eps_i)\cdot d_i(v)/2 - 1$. The time for computing this edge orientation is $\TBO(1/\eps_i)$ \LOCAL rounds on graph $G_i$. Node $v$ can therefore build at least $\frac{1-\eps_i}{4}\cdot d_i(v)-1$ pairs of outgoing edges. Consider one such pair of edges $\set{v,u}$ and $\set{v,w}$. Node $v$ splits off the two edges and creates a new virtual edge between nodes $u$ and $w$ instead. Note that because the two edges are outgoing, they can only be paired up by $v$ and not by either $u$ or $w$. By splitting off the two edges, the degree of $v$ drops by $2$. The type of the new edge $\set{u,w}$ depends on the types of edges $\set{v,u}$ and $\set{v,w}$. If edges $\set{v,u}$ and $\set{v,w}$ have the same type, then the type of $\set{u,w}$ is $\typeO$, otherwise, the type of edge $\set{u,w}$ is $\typeC$. In this way, any proper $\red$/$\blue$-labeling of the half-edges of $\set{u,w}$ can be translated to a proper $\red$/$\blue$-labeling of the half edges of $\set{v,u}$ and $\set{v,w}$ such that a) the two half-edges incident to $v$ are colored with different colors and b) the half-edges adjacent to $u$ and $w$ are colored in the same way as in $\set{u,w}$. From a solution to \Cref{eq:balancedsplitting3} in $G_{i+1}$, we can therefore map back to a solution of \Cref{eq:balancedsplitting3} in $G_i$. If we split off a virtual edge for each of the at least $\frac{1-\eps_i}{4}\cdot d_i(v)-1$ pairs of outgoing edges of $v$, the degree of $v$ drops by at least $\frac{1-\eps_i}{2}\cdot d_i(v)-2$ and we therefore have
	\begin{equation}\label{eq:degreedrop}
		d_{i+1}(v) \leq d_i(v) - \frac{1-\eps_i}{2}\cdot d_i(v) + 2 = \frac{1+\eps_i}{2}\cdot d_i(v) + 2
		\quad\text{and thus}\quad
		\Delta_{i+1} \leq \frac{1+\eps_i}{2}\cdot \Delta_i + 2.
	\end{equation}
	As a result, we can upper bound the maximum degree $\Delta_i$ of $G_i$ by
	\begin{equation}\label{eq:Deltabound}
		\Delta_i\leq \frac{\Delta}{2^i}\cdot e^{\sum_{j=0}^{i-1}\eps_i}+6
		\leq \frac{\Delta}{2^i} \cdot e^{\frac{1}{4\sqrt{\Delta}}\cdot \sum_{j=-\infty}^{k-1}2^{j/2}} + 6
		= \frac{\Delta}{2^i}\cdot e^{\frac{1}{4\sqrt{\Delta}}\cdot e^{(k-1)/2}\cdot (\sqrt{2}+1)} + 6
		< \frac{2\Delta}{2^i} + 6.
	\end{equation}
	We show the first inequality by induction on $i$. The remaining inequalities follow by plugging in the definitions of $\eps_i$ and $k$. The first inequality holds for $i=0$ since $\Delta_0=\Delta \leq e^0\Delta+6$. For $i\geq0$, \Cref{eq:degreedrop} together with the induction hypothesis implies that
	\[
	\Delta_{i+1} \leq \frac{1+\eps_i}{2}\cdot\Delta_i + 2 \leq
	\frac{e^{\eps_i}}{2}\cdot \frac{\Delta}{2^i}\cdot e^{\sum_{j=0}^{i-1}\eps_i}+\frac{1+\eps_i}{2}\cdot 6 + 2 
	< \frac{\Delta}{2^{i+1}}\cdot e^{\sum_{j=0}^{i}\eps_i} + 6.
	\]
	The last inequality follows because $\eps_i\leq 1/4$ and thus $(1+\eps_i)/2\cdot 6 + 2 < 6$.
	Because some of the edges of $G_{i+1}$ are virtual edges consisting of two edges of $G_i$, running one round of the \LOCAL model on $G_{i+1}$ can be done in $2$ rounds in the \LOCAL model on $G_i$. 
	
	The total time $T_k$ (in \LOCAL on $G$) to perform the $k=\lceil\log_2\Delta\rceil$ degree reduction steps is therefore
	\begin{eqnarray*}
		T_k & \leq & \sum_{i=0}^{k-1}2^i\cdot \TBO\left(\frac{1}{\eps_i}\right)\\
		& \leq & \alpha\cdot\sum_{i=0}^{k-1}2^i\cdot \frac{1}{\eps_i}\cdot \log\Delta\cdot\log^{1.71}\log\Delta\cdot\TSO\\
		& \leq & \alpha\cdot \sum_{i=0}^{k-1}2^{i/2}\cdot 4\sqrt{\Delta}\cdot \log\Delta\cdot\log^{1.71}\log\Delta\cdot\TSO\\
		& = & O\left(2^{(k-1)/2}\cdot 4\sqrt{\Delta}\cdot \log\Delta\cdot\log^{1.71}\log\Delta\cdot\TSO\right)\ =\ \TBO(\Delta).
	\end{eqnarray*}
	
	At the end, we get graph $G_k$ with maximum degree  $\Delta_k< 2\Delta/2^k+6\leq 8$ (because $k\geq \log_2\Delta$). We thus have $\Delta_k\leq 7$. Because $k=\lceil\log_2\Delta\rceil$, the virtual edges in $G_k$ consist of at most $2^k=O(\Delta)$ edges of $G$. Each round on $G_k$ can therefore be simulated in $O(\Delta)$ rounds on $G$.
	
	We finish the computation by first computing two additional graphs $G_{k+1}$ and $G_{k+2}$ in a similar way. In both cases, we compute an orientation in which every node of degree at least $5$ has outdegree at least $2$ and we then use this orientation in the same way as before to reduce the degree of each node of degree at least 5 by two. By using Lemma 2.6 in \cite{Ghaffari2020ImprovedDD}, such an orientation can be computed in time $\TSO$ on $G_k$ (for getting $G_{k+1}$) and on $G_{k+1}$ (for getting $G_{k+2}$). The round complexity on $G$ for those two steps is therefore $O(\Delta)\cdot\TSO=\TBO(\Delta)$. Note that in each step, the length of the virtual edges grows by at most a factor $2$ and therefore in $G_k$, $G_{k+1}$, and $G_{k+2}$, one \LOCAL round can be simulated in $O(\Delta)$ \LOCAL rounds on $G$. We now have a graph $G_{k+2}$ of maximum degree at most $4$ such that finding a valid $\red$/$\blue$-labeling of the half-edges that satisfies \Cref{eq:balancedsplitting3} on $G$ can be solved by finding such a  $\red$/$\blue$-labeling on $G_{k+2}$. The degree of this graph is now however small enough so that this can be achieved in time $\TSO$ on $G_{k+2}$ by applying \Cref{lem:reduce_to_SO}. This concludes the proof of the lemma.
\end{proof}

\Cref{lemma:balancedsplitting1} already proves the first part \Cref{eq:balancedsplitting1} of \Cref{thm:balancedsplitting}. The next lemma shows how to also prove the second part \Cref{eq:balancedsplitting2} of the theorem.

\begin{lemma}\label{lemma:balancedsplitting2}
	Let $G=(V,E)$ be an $n$-node multigraph without selfloops and assume that every edge $e\in E$ is of type $\typeC$. There is a $\TBO(\Delta)$-round \LOCAL algorithm for computing a red-blue coloring of the half-edges of $G$ such that for every node $v\in V$ of degree $d(v)$, it holds that
	\begin{equation}\label{eq:perfectsplit_rounddown}
		\dred(v) \in \set{\left\lfloor\frac{d(v)}{2}\right\rfloor, \left\lfloor\frac{d(v)}{2}\right\rfloor+1}
		\quad\text{and}\quad
		\dblue(v) \in \set{\left\lceil\frac{d(v)}{2}\right\rceil-1, \left\lceil\frac{d(v)}{2}\right\rceil}.
	\end{equation}
\end{lemma}
\begin{proof}
	Note that for odd-degree nodes, the condition of \Cref{eq:perfectsplit_rounddown} is implied by the condition of \Cref{lemma:balancedsplitting1}. Our approach therefore is to reduce the problem on general graphs to the problem on graphs with only odd-degree nodes and to then use \Cref{lemma:balancedsplitting1}. For this, we use the fact that all edges of $G$ are of type $\typeC$ (i.e., we need to compute a red/blue-coloring of the edges of $G$). The reduction consists of two steps.
	
	Let $V_E\subseteq V$ be the set of nodes that have an even degree in $G$ and let $G[V_E]$ be the subgraph of $G$ induced by $V_E$. As the first step, we compute a  maximal matching $\mathcal{M}$ of $G[V_E]$. This can be done in $O(\Delta+\log^* n)$ rounds~(e.g., \cite{Kuhn2009}). In the end, the edges of $\mathcal{M}$ will be labeled $\red$ (i.e., both their half-edges are labeled $\red$). We now however first remove the edges of $\mathcal{M}$ from the graph. Let $V_M\subseteq V_E$ be the set of nodes that are incident to an edge in the matching $\mathcal{M}$, let $E':=E\setminus \mathcal{M}$, and let $G'=(V,E')$. Note that in graph $G'$, all nodes in $V_M$ now have an odd degree (because they lost exactly one of their edges) and the only nodes of even degree are the nodes $V_E':=V_E\setminus V_M$. Because of the maximality of $\mathcal{M}$, the nodes in $V_E'$ form an independent set of $G'$. Therefore, the nodes in $V_E'$ can pair up their edges and create a virtual edge of type $\typeO$ out of each of those pairs. Assume that $v\in V_E'$ and that we pair up two edges $\set{v,u}$ and $\set{v,w}$ to create the virtual edge $\set{u,w}$ of type $\typeO$. The half-edges of $\set{u,w}$ correspond to the edges $\set{v,u}$ and $\set{v,w}$ and when computing a valid assignment of labels to $\set{u,w}$, we therefore color one of the two edges of $v$ red and the other one blue. Node $v$ therefore gets exactly the same number of red and blue edges and it therefore clearly satisfies the requirement of \Cref{eq:perfectsplit_rounddown}.
	
	The next graph $G''=(V'',E'')$ is now built as follows. The set of nodes of $G''$ is $V''=V\setminus V_E'=(V\setminus V_E)\cup V_M$ and the set of edges consists of all edges of $E'$ that are not incident to some node in $V_E'$ and it consists of the virtual edges that are created when removing the nodes in $V_E'$. Note that $G''$ can be a multigraph even if $G$ and $G'$ are simple graphs. Note also that $G''$ has only odd-degree nodes. All nodes that have an odd degree in $G$, i.e., the nodes $V\setminus V_E$ have the same degree in $G''$. All nodes in $V_M$ have their original (even) degree minus $1$ (from removing $\mathcal{M}$).
	
	In order to compute the labeling of $G$, we now apply \Cref{lemma:balancedsplitting1} to graph $G''$ (note that one \LOCAL round in $G''$ can be simulated in $2$ \LOCAL rounds in $G$). Let $v\in V''$ be a node of $G''$ and let $d''(v)$ be the (odd) degree of $v$ in $G''$. \Cref{lemma:balancedsplitting1}  implies that the red and blue degree of $v$ in $G''$ are either $\lfloor d''(v)/2\rfloor$ or $\lceil d''(v)/2\rceil = \lfloor d''(v)/2\rfloor+1$. For the nodes $v\in V\setminus V_E$, we have $d(v)=d''(v)$ and thus also $\dred(v),\dblue(v)\in \set{\lfloor d(v)/2\rfloor, \lfloor d(v)/2\rfloor+1}$ as required by \Cref{eq:perfectsplit_rounddown}. For $v\in V_M$, we have $d(v)=d''(v)+1$ and since the edges of $\mathcal{M}$ are colored red, we have $\dred(v)\in \set{\lfloor d''(v)/2\rfloor+1, \lfloor d''(v)/2\rfloor+2} = \set{\lfloor d(v)/2\rfloor, \lfloor d(v)/2\rfloor +1}$. This completes the proof.
\end{proof}

Note that by switching the names of the colors, the algorithm of \Cref{lemma:balancedsplitting2} can also be used to guarantee the following condition.
\begin{equation}\label{eq:perfectsplit_roundup}
	\dred(v) \in \set{\left\lceil\frac{d(v)}{2}\right\rceil-1, \left\lceil\frac{d(v)}{2}\right\rceil}
	\quad\text{and}\quad
	\dblue(v) \in \set{\left\lfloor\frac{d(v)}{2}\right\rfloor, \left\lfloor\frac{d(v)}{2}\right\rfloor+1}.
\end{equation}

\section{Arbitrary Unoriented Degree Splittings}
\label{sec:arbitrary2splitting}

We next prove \Cref{thm:arbitrary2splitting}, i.e., that the problems $\Pi(y)$ (cf.~\Cref{def:arbitrary2splitting}) can be solved in $O(\log\Delta)\cdot\TBO(\Delta)$ rounds. We first recap the definition of the problem family $\Pi(y)$ for $y\in \set{0,\dots,\Delta-1}$. The problem asks for a red-blue coloring of the edges of an $n$-node graph $G=(V,E)$ of maximum degree $\Delta$. An algorithm solves $\Pi(y)$ if the following holds.
\begin{enumerate}
	\item For every node $v\in V$ of degree $\Delta$, the number of red edges is in $\set{y, y+1}$. Just requiring this subproblem is also defined as problem $\Pi_\Delta(y)$.
	\item For every degree $d<\Delta$, there exists a value $x_d$ such that the number of red edges of all nodes of degree $d$ is in $\set{x_d-1,x_d,x_d+1}$.
	\item The value of $x_d$ for $d<\Delta$ must satisfy the
          following conditions. Let $\tau_d:=\frac{d}{2\Delta}$,
          $\tilde{x}_d:=\frac{d}{\Delta}\cdot\left(y+\frac{1}{2}\right)$,
          $\alpha_d:=\left\lfloor\tilde{x}_d\right\rfloor$, and
          $\beta_d:=\tilde{x}_d-\alpha$.
         \item If $\beta_d\leq \tau_d$, $x_d = \alpha_d$, if $\beta_d
           \in (\tau_d, 1-\tau_d)$, $x_d\in \set{\alpha_d,
             \alpha_d+1}$, and if $\beta_d \geq 1 - \tau_d$, $x_d=\alpha_d+1$.
%	\begin{itemize}
%		\item If $\beta_d\leq \tau_d$, $x_d = \alpha_d$.
%		\item If $\beta_d \in (\tau_d, 1-\tau_d)$, $x_d\in \set{\alpha_d, \alpha_d+1}$.
%		\item If $\beta_d \geq 1 - \tau_d$, $x_d=\alpha_d+1$
%	\end{itemize}
\end{enumerate}

Because this will turn out to be more convenient, as a first step, we extend the definition of $\Pi(y)$ to all $y\in \mathbb{Z}$ as follows. For $y\in \mathbb{Z}$, we define
\(
\Pi(y) := \Pi(y \mod (\Delta-1))
\),
where we assume that $y \mod (\Delta-1)$ maps to the range $\set{0,\dots,\Delta-2}$. We note that this changes the definition of $\Pi(\Delta-1)$ to be equal to $\Pi(0)$. Note however that $\Pi(\Delta-1)$ requires all edges to be red and $\Pi(0)$ requires all edges to be blue and up to renaming the colors, the two problems are indeed equivalent. Two problems are generally equivalent if they can be obtained from each other by just switching the name of the two colors. The following lemma makes this formal. 

\begin{lemma}\label{lemma:Piequivalence}
	For any $y \in \mathbb{Z}$, the problems $\Pi(y)$ and $\Pi(\Delta - 1 - y) = \Pi(- y)$ are equivalent.
\end{lemma}
\begin{proof}
	Assume that $y\in \set{0,\dots,\Delta-2}$. In problem $\Pi_\Delta(y)$, the number of red edges of nodes of degree $\Delta$ is required to be in $\set{y,y+1}$. In problem $\Pi(\Delta-1-y)$, the number of red edges of nodes of degree $\Delta$ is required to be in $\set{\Delta-y-1, \Delta-y}$ and the number of blue edges is thus required to be in $\set{y, y+1}$.
	
	Let us also consider nodes of degree $d<\Delta$. We use the variables $\tilde{x}_d$, $x_d$, $\tau_d$, $\alpha_d$, and $\beta_d$ for the problem $\Pi(y)$ and $\tilde{x}_d'$, $x_d'$, $\tau_d'$, $\alpha_d'$, and $\beta_d'$ for the problem $\Pi(\Delta-1-y)$. Note that $\tilde{x}_d:=\frac{d}{\Delta}\cdot\left(y+\frac{1}{2}\right)$ and
	\[
	\tilde{x}_d' = \frac{d}{\Delta}\cdot \left(\Delta-1-y+\frac{1}{2}\right) =
	\frac{d}{\Delta}\cdot\left(\Delta-y-\frac{1}{2}\right) = d-\tilde{x}_d.
	\]
	It is sufficient to show that if $x_d$ is unique, then $x_d'$ is unique and we have $x_d'=d-x_d$. In this case the number of red edges in $\Pi(y)$ is in $\set{x_d-1,x_d,x_d+1}$ and the number of red edges in $\Pi(\Delta-1-y)$ is in $\set{d-x_d-1, d-x_d, d-x_d+1}$ and thus the number of blue edges is in $\set{x_d-1,x_d,x_d+1}$. This implies that the two problems are equivalent. If $x_d$ can take two different values $x_{d,1}$ and $x_{d,2}$, then we show that also $x_d'$ can take two different values and we have $x_{d,1}'=d-x_{d,1}$ and $x_{d,2}'=d-x_{d,2}$. This then again implies that the two problems are equivalent because the allowed settings for the blue edges in $\Pi(\Delta-1-y)$ is the same as the allowed setting for the red edges in $\Pi(y)$. We need to make a case distinction based on whether $\tilde{x}_d$ is an integer or not. 
	\paragraph{The value $\tilde{x}_d$ is an integer.}
	In this case, we have $\alpha_d=\tilde{x}_d$ and thus $\alpha_d'=d-\alpha_d$. We further have $\beta_d=\beta_d'=0$. We therefore have $x_d=\tilde{x}_d=\alpha_d$ and $x_d'=\tilde{x}_d'=\alpha_d'=d-\alpha_d$ as required.
	
	\paragraph{The value $\tilde{x}_d$ is not an integer.}
	In this case, we have $\alpha_d'=\lfloor \tilde{x}'_d\rfloor$ and thus $\alpha_d'=\lfloor d-\tilde{x}_d\rfloor = d-\alpha_d-1$. We further have $\beta_d = \tilde{x}_d-\alpha_d$ and $\beta_d'=d-\tilde{x}_d -\alpha_d' = \alpha_d+1 - \tilde{x}_d=1-\beta_d$. We also clearly have $\tau_d'=\tau_d$ as this value only depends on $d$ and $\Delta$. We make a further case distinction based on the value of $\beta_d$
	\begin{itemize}
		\item If $\beta_d\leq \tau_d$, we have $x_d=\alpha_d$. For problem $\Pi(\Delta-1-y)$, we have $\beta_d'=1-\beta_d$ and thus $\beta_d'\geq 1-\tau_d$. We therefore have $x_d'=\alpha_d'+1=d-\alpha_d=d-x_d$. 
		\item If $\beta_d\geq 1-\tau_d$, we have $x_d=\alpha_d+1$. For problem $\Pi(\Delta-1-y)$, we have $\beta_d'=1-\beta_d$ and thus $\beta_d'\leq \tau_d$. We therefore have $x_d'=\alpha_d'=d-\alpha_d-1=d-x_d$.
		\item Finally, if $\tau_d<\beta_d\leq 1-\tau_d$, we have $x_d\in \set{\alpha_d, \alpha_d+1}$. We further have $\beta_d'=1-\beta_d$ and thus also $\beta_d'\in(\tau_d, 1-\tau_d)$. We therefore have $x_d'\in \set{\alpha_d', \alpha_d'+1} = \set{d-1-\alpha_d', d-\alpha_d'}$. If $x_d=\alpha_d$, then $x_d'=\alpha_d'+1=d-\alpha_d' =d-x_d$ and if $x_d=\alpha_d+1$, then $x_d'=\alpha_d'=d-1-\alpha_d'=d-x_d$.
	\end{itemize}
	Hence, we get that $\Pi(\Delta-1-y)$ is in all cases equivalent to $\Pi(y)$.
	
\end{proof}

We are now ready to show that $\Pi(y)$ can be computed in time $O(\log\Delta)$ steps, where each step is an application of \Cref{lemma:balancedsplitting2}. The high level idea is the following. In the next two technical lemmas, we show that if a solution to $\Pi(2y)$ is given, then we can compute a solution to $\Pi(y)$ in time $\TBO(\Delta)$ (by applying \Cref{lemma:balancedsplitting2}). Depending on the value of $y$, we can further also compute a solution to $\Pi(y)$ if either a solution of $\Pi(2y+1)$ or of $\Pi(2y-1)$ is given. This will be sufficient to show that if we do $k$ consecutive applications of \Cref{lemma:balancedsplitting2}, it is possible to get to $\Pi(y)$ from $\Pi(y')$ for $2^k$ consecutive $y'$-values. We now first show how to get from a solution to $\Pi(2y)$ to a solution to $\Pi(y)$.

\begin{lemma}\label{lemma:pi(2y)}
	Let $G = (V, E)$ be an $n$-node multigraph without selfloops. For all  $y\in \set{0,\dots,\Delta-2}$, given a solution for $\Pi(2y)$ on $G$, we can compute a solution for $\Pi(y)$ in $\TBO(\Delta)$ rounds.
\end{lemma}
\begin{proof}
	We first describe the transformation from $\Pi(2y)$ to $\Pi(y)$ and we afterwards prove that the resulting problem satisfies all the requirements of problem $\Pi(y)$. We distinguish two cases.
	\begin{itemize}
		\item $y \le \frac{\Delta - 2}{2}$:
		\begin{enumerate}
			\item Let $E'$ be the set of red edges in the given solution to $\Pi(2y)$ and let $G'$ be the subgraph of $G$ induced by the edges in $E'$
			\item We apply \Cref{lemma:balancedsplitting2} (with \Cref{eq:perfectsplit_rounddown}) to $G'$. The edges that are colored red in the resulting red-blue labeling of $G'$ remain red and the remaining edges of $E'$ together with the edges in $E\setminus E'$ become blue in the final labeling.
		\end{enumerate}
		\item $y > \frac{\Delta - 2}{2}$:
		\begin{enumerate}
			\item We first switch the colors in the given solution to $\Pi(2y)$
			\item Let $E'$ be the set of red edges in the current solution to $\Pi(2y)$ and let $G'$ be the subgraph of $G$ induced by the edges in $E'$
			\item We apply \Cref{lemma:balancedsplitting2} (with \Cref{eq:perfectsplit_rounddown}) to $G'$. The edges that are colored red in the resulting red-blue labeling of $G'$ remain red and the remaining edges of $E'$ together with the edges in $E\setminus E'$ become blue.
			\item We again switch the colors for the final labeling of the edges.
		\end{enumerate}
	\end{itemize}
	
	For the analysis, we also consider the two cases.
	\begin{itemize}
		\item $y \le \frac{\Delta - 2}{2}$: For a node $v$, let $\dred'(v)$ be its degree in $G'$ (or equivalently, $\dred'(v)$ is the number of red edges in the initial solution to $\Pi(2y)$). We further use $\dred(v)$ to denote the number of red edges after executing the algorithm. We thus need to show that the values of $\dred(v)$ satisfy the requirements of problem $\Pi(y)$. Note that for every $v\in V$, by \Cref{eq:perfectsplit_rounddown}, the number of red edges after executing the algorithm is
		\begin{equation}\label{eq:dredbound1}
			\dred(v) \in \set{\left\lfloor\frac{\dred'(v)}{2}\right\rfloor, \left\lfloor\frac{\dred'(v)}{2}\right\rfloor + 1}.
		\end{equation}
		We first concentrate on nodes $v$ that have degree $\Delta$ in $G$. Since the initial red-blue labeling is a solution to $\Pi(2y)$, we know that for each such node $v$, $\dred'(v)\in \set{2y, 2y+1}$. Note that $\lfloor (2y)/2\rfloor = \lfloor (2y+1)/2\rfloor = y$. We therefore have $\dred(v)\in\set{y, y+1}$ as it is required for nodes of degree $\Delta$ in problem $\Pi_\Delta(y)$ (and thus in problem $\Pi(y)$).
		
		Let us now consider some node $v$ of degree $d<\Delta$ in $G$. We use the following definition from \Cref{def:arbitrary2splitting}:
		\begin{equation}\label{eq:Piy_vars}
			\tilde{x}_d := \frac{d}{\Delta}\cdot\left(y+\frac{1}{2}\right),\quad
			\alpha_d := \lfloor \tilde{x}_d \rfloor,\quad \beta_d := \tilde{x}_d - \alpha_d,\quad\text{and}\quad\tau_d := \frac{d}{2\Delta}.
		\end{equation}
		
		For the problem $\Pi(2y)$, we also define the corresponding value as follows.
		\begin{equation}\label{eq:Pi2y_vars}
			\tilde{x}_d' := \frac{d}{\Delta}\cdot\left(2y+\frac{1}{2}\right),\quad
			\alpha_d' := \lfloor \tilde{x}_d' \rfloor,\quad\text{and}\quad \beta_d' := \tilde{x}_d' - \alpha_d'.
		\end{equation}
		Note that we have
		\begin{equation}
			\label{eq:farsrelation}
			\tilde{x}_d = \alpha_d + \beta_d\quad\text{and}\quad
			\tilde{x}_d' = \alpha_d' + \beta_d' = 2\alpha_d' + 2\beta_d' - \tau_d.
		\end{equation}
		We make a case distinction based on the value of $\beta_d$.
		\begin{itemize}
			\item $2\beta_d < \tau_d$: In this case, we know that $\alpha_d' = 2\alpha-1$ and $\beta_d'=1+2\beta_d-\tau_d\geq 1-\tau_d$. Note that this implies that $x_d' = \alpha_d'+1=2\alpha_d$. Since $\beta_d<\tau$, we further know that $x_d = \alpha_d$. We therefore have $\dred'(v)\in\set{2\alpha_d-1,2\alpha_d, 2\alpha_d+1}$ and we need $\dred(v)\in \set{\alpha_d-1, \alpha_d, \alpha_d+1}$. By applying \Cref{eq:dredbound1}, we get that $\dred(v) \geq \left\lfloor \frac{2\alpha_d-1}{2}\right\rfloor = \alpha_d-1$ and $\dred(v)\leq \left\lfloor \frac{2\alpha_d+1}{2}\right\rfloor +1 = \alpha_d+1$.
			\item $0\leq 2\beta_d - \tau < 1$. In this case, we know that $\alpha_d' = 2\alpha_d$ and $\beta_d' = 2\beta_d -\tau_d$. We make a further case distinction
			\begin{itemize}
				\item $\beta_d \leq \tau_d$: In this case, we have $\beta_d' = 2\beta_d - \tau_d \leq \tau_d$. We therefore know that $x_d' =\alpha_d'=2\alpha_d$ and we have $x_d =\alpha_d$. We therefore have $\dred'(v)\in \set{2\alpha_d-1,2\alpha_d, 2\alpha_d+1}$ and we need $\dred(v)\in\set{\alpha_d-1, \alpha_d, \alpha_d+1}$. By \Cref{eq:dredbound1}, we have $\dred(v)\geq \lfloor (2\alpha_d - 1)/2\rfloor = \alpha_d-1$. We similarly have $\dred(v)\leq \lfloor (2\alpha_d + 1)/2\rfloor+1=\alpha_d+1$.
				\item $\tau_d<\beta_d<1-\tau_d$: Note that since $\beta_d$ is between $\tau_d$ and $1-\tau_d$, in this case, we only need to show that either $x_d=\alpha_d$ and thus $\dred(v)\in\set{\alpha_d-1,\alpha_d,\alpha_d+1}$ or $x_d = \alpha_d+1$ and thus $\dred(v)\in \set{\alpha_d,\alpha_d+1, \alpha_d+2}$. Independently of what the value of  $\beta_d'$ is, we know that $x_d' = 2\alpha_d$ or $x_d' = 2\alpha_d+1$. In the first case, we get $\dred(v)\geq \lfloor (2\alpha_d-1)/2\rfloor = \alpha_d-1$ and $\dred(v)\leq \lfloor(2\alpha_d+1)/2\rfloor+1=\alpha_d+1$. In the second case, we get $\dred(v)\geq \lfloor (2\alpha_d)/2\rfloor = \alpha_d$ and $\dred(v)\leq \lfloor(2\alpha_d+2)/2\rfloor+1=\alpha_d+2$. We therefore get a valid solution for $\Pi(y)$ in either case.
				\item $\beta_d\geq 1-\tau_d$: In this case, we have $\beta_d'=2\beta_d-\tau_d \geq 2 -3\tau_d > 1-\tau_d$ (because $\tau_d<1/2$). We therefore know that $x_d'=\alpha_d'+1=2\alpha_d+1$ and thus $\dred'(v)\in\set{2\alpha_d, 2\alpha_d+1, 2\alpha_d+2}$. As above, in this case, we get $\dred(v)\in\set{\alpha_d, \alpha_d+1, \alpha_d+2}$, which is what we need because $\beta_d\geq 1-\tau_d$.
			\end{itemize}
			\item $2\beta_d-\tau_d \geq 1$: In this last case, we have $\beta_d'=2\beta_d-\tau_d-1$ and $\alpha_d'=2\alpha_d+1$. We cannot bound $\beta_d'$ in this case and we therefore only know that $x_d'\in\set{2\alpha_d+1, 2\alpha_d+2}$ and thus $\dred'(v)\in\set{2\alpha_d,\dots,2\alpha_d+3}$. This implies that $\dred(v) \geq \lfloor(2\alpha_d)/2\rfloor=\alpha_d$ and $\dred(v)\leq \lfloor (2\alpha_d+3)/2\rfloor+1 = \alpha_d+2$.
			Since $2\beta_d\geq 1+\tau_d$, we have $\beta\geq (1+\tau_d)/2 >\tau_d$ and it is thus fine if $x_d=\alpha_d+1$ and thus $\dred(v)\in\set{\alpha_d,\alpha_d+1,\alpha_d+2}$.
		\end{itemize}
		
		\item $y > \frac{\Delta - 2}{2}$: Note that $2y\geq \Delta-1$. Assume that $2y=\Delta-1 + k$ for some $k\in \set{0,\dots,\Delta-3}$. By definition, we have $\Pi(2y)=\Pi(k)$. As step (1), we switch the colors, we move from $\Pi(k)$ to the equivalent $\Pi(\Delta-1-k)$ (cf.\ \Cref{lemma:Piequivalence}). That is, before defining $G'$ and applying degree splitting part to $G'$, the set of red edges satisfies the conditions of $\Pi(\Delta-1-k)$. Note that $\Delta-1-k$ has the same parity as $\Delta-1+k$ and since $\Delta-1+k=2y$, we know that $\Delta-1-k$ is even. We can therefore write the value as $\Delta-1-k=2z$ for some $0\leq z\leq (\Delta-2)/2$. Steps (2) and (3) are now the same as steps (1) and (2) in the first case. With the same analysis as before, we can therefore conclude that the red-blue labeling after step (3) is a valid solution to $\Pi(z)$. Since we again switch the colors in step (4), by \Cref{lemma:Piequivalence}, at the end, the red-blue labeling satisfies the conditions of $\Pi(\Delta-1-z)$. We have
		\[
		\Delta - 1 - z = \Delta - 1 - \frac{\Delta - 1 - k}{2} = \frac{\Delta -1 + k}{2} = y.
		\]
	\end{itemize}
	This concludes the proof of the lemma.
\end{proof}

We next show that a solution to $\Pi(2y+1)$ or to $\Pi(2y-1)$ (depending on the value of $y$) can be also used to compute a solution to $\Pi(y)$. The algorithm and also its proof are similar to the case of \Cref{lemma:pi(2y)}. 

\begin{lemma}\label{lemma:pi(2yp1)}
	Let $G = (V, E)$ be an $n$-node multigraph without selfloops. For all $0\leq y\leq (\Delta-3)/2$, given a solution for $\Pi(2y+1)$ on $G$, we can compute a solution for $\Pi(y)$ in $\TBO(\Delta)$ rounds. Further, for all $(\Delta-3)/2<y\leq \Delta-2$, given a solution for $\Pi(2y-1)$ on $G$, we can compute a solution for $\Pi(y)$ in $\TBO(\Delta)$ rounds.
\end{lemma}
\begin{proof}
	We again first describe the transformation from $\Pi(2y+1)$ to $\Pi(y)$ and we afterwards prove that the resulting problem satisfies all the requirements of problem $\Pi(y)$. W.l.o.g., we assume that $y\in \set{0,\dots,\Delta-2}$. We distinguish two cases.
	\begin{itemize}
		\item $y \le \frac{\Delta - 3}{2}$:
		\begin{enumerate}
			\item Let $E'$ be the set of red edges in the given solution to $\Pi(2y+1)$ and let $G'$ be the subgraph of $G$ induced by the edges in $E'$
			\item We apply \Cref{lemma:balancedsplitting2} (with \Cref{eq:perfectsplit_roundup}) to $G'$. The edges that are colored red in the resulting red-blue labeling of $G'$ remain red and the remaining edges of $E'$ together with the edges in $E\setminus E'$ become blue in the final labeling.
		\end{enumerate}
		\item $y > \frac{\Delta - 3}{2}$:
		\begin{enumerate}
			\item We first switch the colors in the given solution to $\Pi(2y+1)$
			\item Let $E'$ be the set of red edges in the current solution to $\Pi(2y+1)$ and let $G'$ be the subgraph of $G$ induced by the edges in $E'$
			\item We apply \Cref{lemma:balancedsplitting2} (with \Cref{eq:perfectsplit_roundup}) to $G'$. The edges that are colored red in the resulting red-blue labeling of $G'$ remain red and the remaining edges of $E'$ together with the edges in $E\setminus E'$ become blue.
			\item We again switch the colors for the final labeling of the edges.
		\end{enumerate}
	\end{itemize}
	
	For the analysis, we also consider the two cases.
	\begin{itemize}
		\item $y \le \frac{\Delta - 3}{2}$: For a node $v$, let $\dred'(v)$ be its degree in $G'$ (or equivalently, $\dred'(v)$ is the number of red edges in the initial solution to $\Pi(2y)$). We further use $\dred(v)$ to denote the number of red edges after executing the algorithm. We thus need to show that the values of $\dred(v)$ satisfy the requirements of problem $\Pi(y)$. Note that for every $v\in V$, by \Cref{eq:perfectsplit_roundup}, the number of red edges after executing the algorithm is
		\begin{equation}\label{eq:dredbound2}
			\dred(v) \in \set{\left\lceil\frac{\dred'(v)}{2}\right\rceil-1, \left\lceil\frac{\dred'(v)}{2}\right\rceil}.
		\end{equation}
		We first concentrate on nodes $v$ that have degree $\Delta$ in $G$. Since the initial red-blue labeling is a solution to $\Pi(2y+1)$, we know that for each such node $v$, $\dred'(v)\in \set{2y+1, 2y+2}$. Note that $\lceil (2y+1)/2\rceil = \lceil (2y+2)/2\rceil = y+1$. We therefore have $\dred(v)\in\set{y, y+1}$ as it is required for nodes of degree $\Delta$ in problem $\Pi_\Delta(y)$ (and thus in problem $\Pi(y)$).
		
		Let us now consider some node $v$ of degree $d<\Delta$ in $G$. As in the proof of \Cref{lemma:pi(2y)}, we use the following definition from \Cref{def:arbitrary2splitting}:
		\begin{equation}\label{eq:Piy_vars2}
			\tilde{x}_d := \frac{d}{\Delta}\cdot\left(y+\frac{1}{2}\right),\quad
			\alpha_d := \lfloor \tilde{x}_d \rfloor,\quad \beta_d := \tilde{x}_d - \alpha_d,\quad\text{and}\quad\tau_d := \frac{d}{2\Delta}.
		\end{equation}
		
		For the problem $\Pi(2y+1)$, we define the corresponding values as follows.
		\begin{equation}\label{eq:Pi2yp1_vars}
			\tilde{x}_d' := \frac{d}{\Delta}\cdot\left(2y+\frac{3}{2}\right),\quad
			\alpha_d' := \lfloor \tilde{x}_d' \rfloor,\quad\text{and}\quad \beta_d' := \tilde{x}_d' - \alpha_d'.
		\end{equation}
		Note that we have
		\begin{equation}
			\label{eq:farsrelation2}
			\tilde{x}_d = \alpha_d + \beta_d\quad\text{and}\quad
			\tilde{x}_d' = \alpha_d' + \beta_d' = 2\alpha_d' + 2\beta_d' + \tau_d.
		\end{equation}
		We again make a case distinction based on the value of $\beta_d$.
		\begin{itemize}
			\item $2\beta_d +\tau_d < 1$: In this case, we know that $\alpha_d' = 2\alpha$ and $\beta_d'=2\beta_d+\tau_d\geq 1-\tau_d$. Note that this implies that $x_d' \in \set{\alpha_d',\alpha_d'+1}=\set{2\alpha_d,2\alpha_d+1}$ and thus $\dred'(v)\in\set{2\alpha_d-1,\dots,2\alpha_d+2}$. We thus have $\dred(v)\geq \lceil(2\alpha_d-1)/2\rceil-1 = \alpha_d-1$ and $\dred(v)\leq\lceil(2\alpha_d+2)/2\rceil = \alpha_d+1$. This implies that we need $x_d=\alpha_d$, which is fine as long as $\beta_d<1-\tau_d$. This however holds because we have $\beta_d<(1-\tau_d)/2 <1/2<1-\tau_d$.
			\item $1\leq 2\beta_d + \tau_d < 2$. In this case, we know that $\alpha_d' = 2\alpha_d+1$ and $\beta_d' = 2\beta_d +\tau_d-1$. We make a further case distinction.
			\begin{itemize}
				\item $\beta_d \leq \tau_d$: In this case, we have $\beta_d' \leq 2\tau_d + \tau_d-1<\tau_d$ because $\tau_d<1/2$. We therefore know that $x_d' =\alpha_d'=2\alpha_d+1$ and we need that $x_d =\alpha_d$. We thus have $\dred'(v)\in \set{2\alpha_d,2\alpha_d+1, 2\alpha_d+2}$ and we need $\dred(v)\in\set{\alpha_d-1, \alpha_d, \alpha_d+1}$. By \Cref{eq:dredbound2}, we have $\dred(v)\geq \lceil (2\alpha_d)/2\rceil-1 = \alpha_d-1$. We similarly have $\dred(v)\leq \lceil (2\alpha_d + 2)/2\rceil=\alpha_d+1$.
				\item $\tau_d<\beta_d<1-\tau_d$: Note that since $\beta_d$ is between $\tau_d$ and $1-\tau_d$, in this case, we only need to show that either $x_d=\alpha_d$ and thus $\dred(v)\in\set{\alpha_d-1,\alpha_d,\alpha_d+1}$ or $x_d = \alpha_d+1$ and thus $\dred(v)\in \set{\alpha_d,\alpha_d+1, \alpha_d+2}$. Independently of what the value of  $\beta_d'$ is, we know that $x_d' = 2\alpha_d+1$ or $x_d' = 2\alpha_d+2$. In the first case, we get $\dred(v)\geq \lceil (2\alpha_d)/2\rceil-1 = \alpha_d-1$ and $\dred(v)\leq \lceil(2\alpha_d+2)/2\rceil=\alpha_d+1$. In the second case, we get $\dred(v)\geq \lceil (2\alpha_d+1)/2\rceil-1 = \alpha_d$ and $\dred(v)\leq \lceil(2\alpha_d+3)/2\rceil=\alpha_d+2$. We therefore get a valid solution for $\Pi(y)$ in either case.
				\item $\beta_d\geq 1-\tau_d$: In this case, we have $\beta_d'\geq 2\beta_d-\tau_d -1> 1-\tau_d$ (because $\beta_d<1$). We therefore know that $x_d'=\alpha_d'+1=2\alpha_d+2$ and thus $\dred'(v)\in\set{2\alpha_d+1, 2\alpha_d+2, 2\alpha_d+3}$. As above, in this case, we get $\dred(v)\in\set{\alpha_d, \alpha_d+1, \alpha_d+2}$, which is what we need because $\beta_d\geq 1-\tau_d$.
			\end{itemize}
			\item $2\beta_d+\tau_d \geq 2$: In this last case, we have $\alpha_d'=2\alpha_d+2$ and $\beta_d'=2\beta_d+\tau_d-2<\tau_d$ (because $\beta_d<1$). We further have $\beta_d\geq (2-\tau_d)/2=1-\tau_d/2>1-\tau_d$. We therefore know that $\dred'(v)\in\set{\alpha_d'-1,\alpha_d',\alpha_d'+1}=\set{2\alpha_d+1,2\alpha_d+2,2\alpha_d+3}$ and we need $\dred(v)\in\set{\alpha_d, \alpha_d+1,\alpha_d+2}$. By using \Cref{eq:dredbound2}, we get $\dred(v)\geq \lceil(2\alpha_d+1)\rceil -1 = \alpha_d$ and $\dred(v)\leq \lceil(2\alpha_d+3)\rceil = \alpha_d+1$ as needed.
		\end{itemize}
		
		\item $y \ge \frac{\Delta+1}{2}$: Note that since $y\geq (\Delta+1)/2$, we have $2y-1\geq \Delta$. Assume that $2y-1=\Delta-1 + k$ for some $k\in \set{1,\dots,\Delta-4}$. By definition, we have $\Pi(2y-1)=\Pi(k)$. As step (1), we switch the colors, we move from $\Pi(k)$ to the equivalent $\Pi(\Delta-1-k)$ (cf.\ \Cref{lemma:Piequivalence}). That is, before defining $G'$ and applying degree splitting part to $G'$, the set of red edges satisfies the conditions of $\Pi(\Delta-1-k)$. Note that $\Delta-1-k$ has the same parity as $\Delta-1+k$ and since $\Delta-1+k=2y-1$, we know that $\Delta-1-k$ is odd. We can therefore write the value as $\Delta-1-k=2z+1$ for some $0\leq z\leq (\Delta-3)/2$. Steps (2) and (3) are now the same as steps (1) and (2) in the first case. With the same analysis as before, we can therefore conclude that the red-blue labeling after step (3) is a valid solution to $\Pi(z)$. Since we again switch the colors in step (4), by \Cref{lemma:Piequivalence}, at the end, the red-blue labeling satisfies the conditions of $\Pi(\Delta-1-z)$. We have
		\[
		\Delta - 1 - z = \Delta - 1 - \frac{\Delta - 2 - k}{2} = \frac{\Delta + k}{2} = y.
		\]
	\end{itemize}
	This concludes the proof of the lemma.
\end{proof}

We are now ready to prove \Cref{thm:arbitrary2splitting}, i.e., that for every $y\in \set{0,\dots,\Delta-2}$, problem $\Pi(y)$ can be computed in time $O(\log\Delta)\cdot\TBO(\Delta)$.

\begin{proof}[\textbf{Proof of \Cref{thm:arbitrary2splitting}.}]
	We first note that for some values of $y$, $\Pi(y)$ is straightforward to obtain. First, $\Pi(0)$ can be solved in $0$ rounds by just coloring all edges blue. In this scenario, $\Pi_{\Delta}(0)$ is satisfied and for all vertices with degree \(d < \Delta\), \(\tilde{x}_d = \frac{d}{2\Delta} < \frac{1}{2}\) leads to \(x_d = 0\), thereby satisfying \(\Pi(0)\). Further, by using \Cref{lemma:balancedsplitting2} (i.e., the second part claim of \Cref{thm:balancedsplitting}), we can compute a red-blue coloring of the edges that satisfies \Cref{eq:perfectsplit_rounddown} (or one that satisfies \Cref{eq:perfectsplit_roundup}). This approach directly satisfies \(\Pi_{\Delta}(y)\) for \(y = \left\lfloor\frac{\Delta}{2}\right\rfloor\). For vertices of degree \(d < \Delta\), the value of \(\tilde{x}_d\) is determined as follows, based on whether \(\Delta\) is odd or even:
	
	\begin{itemize}
		\item For even \(\Delta\): \(\tilde{x}_d = \frac{d}{\Delta}\left( \frac{\Delta}{2} + \frac{1}{2}\right) = \frac{d}{2} + \frac{d}{2\Delta}\). If \(d\) is odd, then \(x_d = \frac{d - 1}{2}\) ($\frac{d}{2\Delta} < \frac{1}{2}$), and the values \(\left\lfloor\frac{d}{2}\right\rfloor\) and \(\left\lfloor\frac{d}{2}\right\rfloor + 1\) fall within the set \(\left\{ \frac{d - 3}{2}, \frac{d -1}{2}, \frac{d + 1}{2} \right\}\). For even \(d\), \(x_d = \frac{d}{2}\), with \(\left\lfloor\frac{d}{2}\right\rfloor\) and \(\left\lfloor\frac{d}{2}\right\rfloor + 1\) belonging to the set \(\left\{ \frac{d}{2} - 1, \frac{d}{2}, \frac{d}{2} + 1 \right\}\).
		\item For odd \(\Delta\): \(\tilde{x}_d = \frac{d}{\Delta}\cdot\frac{\Delta}{2} = \frac{d}{2}\). If \(d\) is odd, then \(x_d = \frac{d - 1}{2}\) and if $d$ is even \(d\), \(x_d = \frac{d}{2}\), which are the two cases that were already covered in the previous case.
	\end{itemize}
	
	Hence, for $y=\big\lfloor\frac{\Delta}{2}\big\rfloor$, problem $\Pi(y)$ follows from \Cref{eq:perfectsplit_rounddown}. By \Cref{lemma:Piequivalence}, we further know that problems $\Pi\left(\big\lfloor\frac{\Delta}{2}\big\rfloor\right)$ and $\Pi\left(\big\lceil\frac{\Delta}{2}\big\rceil-1\right)$ are equivalent and therefore also $\Pi\left(\big\lceil\frac{\Delta}{2}\big\rceil-1\right)$ can be solved in time $\TBO(\Delta)$.
	
	In the case, where $y$ does not fit one of the straightforward scenarios, we proceed as follows. For some integer $k$, let $y_1, \dots, y_k$ be $k$ consecutive integers such that for all $i\in \set{1,\dots,k-1}$, $y_{i+1}=y_i+1$. Assume that we want to solve one of the problems $\Pi(y_1),\dots,\Pi(y_k)$. If for some $i$, $y_i\mod (\Delta-1)\in \set{0, \lceil \Delta/2\rceil -1, \lfloor \Delta/2\rfloor}$, we know that the task of solving one of the $\Pi(y_i)$ can be completed in time at most $\TBO(\Delta)$. Assume therefore that
	\begin{equation}
		\label{eq:notyetsolved}
		\forall i\,:\, y_i\mod (\Delta-1)\not\in \set{0, \left\lceil \frac{\Delta}{2}\right\rceil -1, \left\lfloor\frac{\Delta}{2}\right\rfloor}.
	\end{equation}
	We show that in this case, one can construct a sequence $y_1',\dots,y_{2k}'$ of $2k$ consecutive integers (again $y_{i+1}'=y_{i}'+1$) such that given a solution to $\Pi(y_j')$ for some $j\in \set{1,\dots,2k}$, it is possible to solve one of the problems $\Pi(y_i)$ for $i\in \set{1,\dots,k}$ in time $\TBO(\Delta)$. Note that if \Cref{eq:notyetsolved} holds, then for all $i\in \set{1,\dots,k}$, it either holds that $y_i \mod (\Delta-1) \le \frac{\Delta - 3}{2}$  or $y_i \mod (\Delta-1) \ge \frac{\Delta + 1}{2}$. We consider the two cases separately.
	\begin{itemize}
		\item If $y_i \mod (\Delta-1) \le \frac{\Delta - 3}{2}$ for all $i$, then \Cref{lemma:pi(2y),lemma:pi(2yp1)} imply that for $i\in \set{1,\dots,k}$, we can compute $\Pi(y_i)$ in time $\TBO(\Delta)$ when given a solution to $\Pi(2y_i)$ or $\Pi(2y_i+1)$. For $y_1'=2y_1$ and $y_{2k}'=2y_k+1$, given a solution to one of the problems $y_1',\dots,y_{2k}'$, we can therefore compute one of the problems $\Pi(y_i)$ for $i\in \set{1,\dots,k}$ in time $\TBO(\Delta)$.
		\item If $y_i \mod (\Delta-1) \ge \frac{\Delta + 1}{2}$ for all $i$, then \Cref{lemma:pi(2y),lemma:pi(2yp1)} imply that for $i\in \set{1,\dots,k}$, we can compute $\Pi(y_i)$ in time $\TBO(\Delta)$ when given a solution to $\Pi(2y_i-1)$ or $\Pi(2y_i)$. For $y_1'=2y_1-1$ and $y_{2k}'=2y_k$, given a solution to one of the problems $y_1',\dots,y_{2k}'$, we can therefore compute one of the problems $\Pi(y_i)$ for $i\in \set{1,\dots,k}$ in time $\TBO(\Delta)$.
	\end{itemize}
	The theorem now follows by repeating this process until the obtained sequence contains one of the problems in $\set{\Pi(0), \Pi\big(\lceil \Delta/2\rceil -1\big), \Pi\big(\lfloor \Delta/2\rfloor\big)}$. As the sequence of problems that suffice to solve $\Pi(y)$ grows by a factor of $2$ in each step, we only need to repeat the process at most $O(\log\Delta)$ times. Hence, $\Pi(y)$ can be computed in at most $O(\log\Delta)\cdot\TBO(\Delta)$ rounds.
\end{proof}

\section{Unbalanced Orientations}
\label{sec:unbalanced}

The goal of this section is to prove \Cref{thm:orientationsUpper,thm:orientationsLower}. We start with \Cref{thm:orientationsUpper}. 

\begin{proof}[\textbf{Proof of \Cref{thm:orientationsUpper}.}]
	Recall that for a graph $G=(V,E)$ and $\rho_1,\rho_2\in[0,1/2]$ such that $\rho_1+\rho_2\geq 1/2$, we need to show that there is a $\TLLL(\Delta^2)$-round algorithm to compute an edge orientation in which every $v\in V$ has outdegree at most $\rho_1\cdot\deg(v) + O(\sqrt{\Delta\log\Delta})$ or indegree at most $\rho_2\cdot\deg(v)+O(\sqrt{\Delta\log\Delta})$. To prove this, we show that the construction of such an edge orientation can be phrased directly as an LLL problem with polynomial criterion. To see this, we define parameters
	\[
	\eta := \frac{1/2 - \rho_2}{1-\rho_1 - \rho_2}\quad\text{and}\quad
	\nu := \rho_1 + \rho_2 - \frac{1}{2}.
	\]
	We define the following (independent) $0$/$1$-random variables: A random variable $X_v\in \set{0,1}$ with $\Pr(X_v=1)=\eta$ for every node $v\in V$, as well as random variables $Y_e,Z_e\in \set{0,1}$ with $\Pr(Y_e=1)=\nu$ and $\Pr(Z_e=1)=1/2$ for every edge $e\in E$. We use those random variables to define a random edge orientation of $G$. For every edge $e=\set{u,v}\in E$ with $\mathrm{ID}(u)<\mathrm{ID}(v)$, we define the orientation of $e$ as follows.
	\begin{itemize}
		\item If $X_u=X_v$, $e$ is oriented from $u$ to $v$ if $Z_e=1$ and it is oriented from $v$ to $u$ otherwise.
		\item If $X_u=0$ and $X_v=1$, $e$ is oriented from $u$ to $v$ if $Y_e=0$ and it is oriented from $v$ to $u$ otherwise.
		\item If $X_u=1$ and $X_v=0$, $e$ is oriented from $u$ to $v$ if $Y_e=1$ and it is oriented from $v$ to $u$ otherwise.
	\end{itemize}
	In the following, let $\gamma>0$ be a constant that we will fix to be sufficiently large. For every node $v\in V$, let $d_0(v)$ be the number of neighbors $u$ of $v$ for which $X_u=0$ and let $d_1(v)$ be the number of neighbors $u$ of $v$ for which $X_u=1$. Note that for every node $v$, we have $\E[d_0(v)]=(1-\eta)\dot d(v)$ and $\E[d_1(v)]=\eta\cdot d(v)$. Note that $d_0(v)$ and $d_1(v)$ are both binomially distributed. By \Cref{lemma:chernoff} (see below), for a sufficiently large constant $c>0$, we therefore have
	\begin{equation}\label{eq:vertexsplitting1}
		\Pr\left(
		\left|d_0(v) - (1-\eta)\cdot d(v)\right| \leq c\sqrt{\Delta\log\Delta}
		\right) \leq \frac{1}{\Delta^{\gamma}}
	\end{equation}
	and
	\begin{equation}\label{eq:vertexsplitting2}
		\Pr\left(
		\left|d_1(v) - \eta\cdot d(v)\right| \leq c\sqrt{\Delta\log\Delta}
		\right) \leq \frac{1}{\Delta^{\gamma}}.
	\end{equation}
	In the following, we fix the assignment of variables $X_v$ for all $v\in V$ and only consider the conditional probability space over the random choice of the edge random variables. Let $V_0$ and $V_1$ be the set nodes $v\in V$ that choose $X_v=0$ and $X_v=1$, respectively. We first consider the case of a node in $v\in V_0$. Let $u$ be a neighbor of node $v$ and w.l.o.g., assume that $\mathrm{ID}(u)<\mathrm{ID}(v)$. Let $e=\set{u,v}$ be an edge between $u$ and $v$. Edge $e$ is an incoming edge for $v$ if either $X_u=0$ and $Z_e=1$ or $X_u=1$ and $Y_e=1$. Let $D_{in}^{(0)}(v)$ be the indegree of a node $v\in V_0$. We have
	\begin{equation}\label{eq:condexp0}
		\E\big[D_{in}^{(0)}(v)\big]  = \frac{1}{2}\cdot d_0(v) + \nu\cdot d_1(v).
	\end{equation}
	We now consider the case of a node in $v\in V_1$. Let $u$ be a neighbor of node $v$ and w.l.o.g., assume that $\mathrm{ID}(u)<\mathrm{ID}(v)$. Let $e=\set{u,v}$ be an edge between $u$ and $v$. Edge $e$ is an outgoing edge for $v$ if which either $X_u=1$ and $Z_e=0$ or $X_u=0$ and and $Y_e=1$. Let $D_{out}^{(1)}(v)$ be the indegree of a node $v\in V_1$. We have
	\begin{equation}\label{eq:condexp1}
		\E\big[D_{out}^{(1)}(v)\big]  = \frac{1}{2}\cdot d_1(v) + \nu\cdot d_0(v).
	\end{equation}
	In the conditional probability space, where the $X_v$-variables are already fixed, the two random variables $D_{in}^{(0)}(v)$ and $D_{out}^{(1)}(v)$ are sums of two (independent) Binomial random variables. We can therefore use \Cref{lemma:chernoff} to show that with probability at least $1-\Delta^{-\gamma}$, for a sufficiently large constant $c>0$, both variables deviate at most $c\sqrt{\Delta\log\Delta}$ (additively) for their expectations as given in \Cref{eq:condexp0,eq:condexp1}. Together with \Cref{eq:vertexsplitting1,eq:vertexsplitting2}, this implies that also in the original non-conditional probability space, $D_{in}^{(0)}(v)$ and $D_{out}^{(1)}(v)$ deviate at most $c\sqrt{\Delta\log\Delta}$ from their expectation  with probability at least $1-\Delta^{-\gamma}$. The expectations of $D_{in}^{(0)}(v)$ and $D_{out}^{(1)}(v)$ can be computed as follows.
	\begin{eqnarray*}
		\E\big[D_{in}^{(0)}(v)\big] & = & \frac{1}{2}\cdot\E\big[d_0(v)\big] + \nu\cdot \E\big[d_1(v)\big]\ =\ 
		\frac{1-\eta}{2}\cdot d(v) + \eta\cdot \nu\cdot d(v)\\
		& = & \left(\frac{1/2-\rho_1}{2(1-\rho_1-\rho_2)} + \frac{(1/2-\rho_2)\cdot(\rho_1+\rho_2-1/2)}{1 - \rho_1  -\rho_2}
		\right)\cdot d(v)\ =\ \rho_2\cdot d(v),\\
		\E\big[D_{out}^{(1)}(v)\big] & = & \frac{1}{2}\cdot\E\big[d_1(v)\big] + \nu\cdot \E\big[d_0(v)\big]\ =\ 
		\frac{\eta}{2}\cdot d(v) + (1-\eta)\cdot \nu\cdot d(v)\\
		& = & \left(\frac{1/2-\rho_2}{2(1-\rho_1-\rho_2)} + \frac{(1/2-\rho_1)\cdot(\rho_1+\rho_2-1/2)}{1 - \rho_1  -\rho_2}
		\right)\cdot d(v)\ =\ \rho_1\cdot d(v).
	\end{eqnarray*}
	With probability $1 - \Delta^{-\gamma'}$ for a $\gamma'$ that we can choose, there therefore exists a constant $c'>0$ such that for every node $v\in V$ that picks $X_v=0$, the indegree is $\rho_2\cdot d(v) \pm c'\cdot\sqrt{\Delta\log\Delta}$ and for every node $v\in V$ that picks $X_v=1$, the outdegree is $\rho_1\cdot d(v) \pm c'\cdot\sqrt{\Delta\log\Delta}$.
	
	We now have a setting to which we can directly apply a constructive LLL algorithm (with polynomial LLL criterion) to obtain an orientation in which every node either has indegree $\rho_2\cdot d(v) \pm c'\cdot\sqrt{\Delta\log\Delta}$ or it has outdegree $\rho_1\cdot d(v) \pm c'\cdot\sqrt{\Delta\log\Delta}$. We create a bad event for every node such that the bad event for node $v$ occurs if $v$ does not satisfy the required criterion. The event for node $v$ only depends on the random variables $Y_e$ and $Z_e$ of the edges of node $v$ and on the random variables $X_v$, and well as $X_u$ for the neighbors of $v$. The events of two nodes $u$ and $v$ are therefore dependent if the two nodes are within distance $2$ in the graph and otherwise they are independent. The dependency degree of the graph of bad events is therefore $O(\Delta^2)$. The desired orientation can therefore be computed in time $\TLLL(\Delta^2)$ as claimed.
\end{proof}

The following simple technical lemma is used in the above proof.

\begin{lemma}\label{lemma:chernoff}
	Let $n$ and $k$ be positive integers such that $k\leq n$ and let $p\in [0,1]$ be a probability (which can depend on $n$ or $k$). Consider a random variable $X\sim\mathrm{Binomial}(k,p)$. For every constant $c>0$, there exists a constant $c'>0$ such that
	\[
	\Pr\left(|X-pk| \leq c'\cdot\sqrt{n\log n}\right) \leq \frac{1}{n^c}.
	\]
\end{lemma}
\begin{proof}
	By a standard Chernoff-Hoeffding bound~\cite{hoeffding}, we know that if $X$ is a sum of $n$ Bernoulli random variables, then for every $t>0$, we have
	\[
	\Pr\big(|X-\E[X]|\geq t\big) \leq 2\cdot e^{-2t^2/n}.
	\]
	The same bound also holds if $X$ is the sum of $k<n$ independent Bernoulli random variables, because we can always add $n-k$ additional dummy variables that are $0$ with probability $1$. The claim of the lemma therefore follows directly by choosing $t=c\sqrt{n\log n}$ for a sufficiently large constant $c>0$.
\end{proof}

We next prove \Cref{thm:orientationsLower} and show that the orientations that are achieved by \Cref{thm:orientationsUpper} are in some sense essentially optimal.

\begin{proof}[\textbf{Proof of \Cref{thm:orientationsLower}.}]
	Suppose there is an algorithm that orients the edges of  $\Delta$-regular trees in $o(\log_{\Delta}n)$ rounds such that each vertex either has outdegree at most $\rho_{1}\cdot\Delta - c\sqrt{\Delta}$ or has indegree at most $\rho_{2}\cdot \Delta - c\sqrt{\Delta}$, where $\rho_{1} + \rho_{2} = \frac{1}{2}$. This implies that for any $\Delta$-regular graph $G$ of girth $\Omega(\log_\Delta n)$, this problem can be also solved in $o(\log_{\Delta}n)$ rounds because nodes can locally not distinguish between the two cases. We will show that for  $\rho_{1} + \rho_{2} = \frac{1}{2}$ and sufficiently large $c$, there are $\Delta$-regular graphs with girth $\Omega(\log_\Delta n)$ for which no such orientation exists.
	
	% Let $G = (V,E)$ be an $n$-node $\Delta$-regular directed high-girth graph where each vertex $v \in V$  either has outdegree at most $\rho_{1}\cdot\Delta - c\sqrt{\Delta}$ or has indegree at most $\rho_{2}\cdot \Delta - c\sqrt{\Delta}$, and $\rho_{1} + \rho_{2} = \frac{1}{2}$.
	
	For this proof, we need to recap some basic spectral graph theory. The adjacency matrix of an undirected $n$-node graph $G=(V,E)$, denoted by $A(G)$, is an $n\times n$ squared matrix, where for each pair of node $u, v\in V$, $A_{u,v} = A_{v,u}= 1$ if $\set{u,v} \in E$ and $A_{u,v} = A_{v,u}= 0$ otherwise. The degree matrix of $G$, denoted by $D(G)$, is a diagonal $n \times n$ matrix, where for each vertex $v \in V$, $D_{v,v} = d(v)$. In addition, the Laplacian matrix of $G$, denoted by $L(G)$, is defined by $L(G) := D(G) - A(G)$. In the following for simplicity we use $A$, $D$, and $L$ whenever we refer to $A(G)$, $D(G)$, and $L(G)$ and $G$ is clear from the context. Let $\mu_1\geq\dots\geq \mu_n$ be the eigenvalues of $A$ and let $0=\lambda_1\leq\dots\leq \lambda_n$ be the eigenvalues of $L$. We note that all those eigenvalues are in $\mathbb{R}$ and if $G$ is $\Delta$-regular, then for each $i$, it holds that $\lambda_i = \Delta - \mu_i$. Friedman~\cite{Friedman2004APO} showed that for the smallest eigenvalue $\mu_n$ of the adjacency matrix of a random $\Delta$-regular graph, w.h.p.\ it holds that $|\mu_n|\leq 2\sqrt{\Delta - 1} + \eps$, for any $\eps > 0$ and thus $\mu_n\geq -2\sqrt{\Delta-1}-\eps\geq -3\sqrt{\Delta}$ for a large enough $\Delta$. We therefore know that the largest eigenvalue $\lambda_n \le \Delta + 3\sqrt{\Delta}$.
	
	The largest eigenvalue of a Laplacian matrix can be computed as the following:
	\[
	\lambda_n := \max_{\vec{x}\in\mathbb{R}^n}\frac{\vec{x}^{\top}L\vec{x}}{\vec{x}^{\top}\vec{x}} = \max_{\vec{x}\in\mathbb{R}^n, \vec{x}^{\top}\vec{x} = 1}\vec{x}^{\top}L\vec{x}.
	\]
	The quadratic form $\vec{x}^{\top}L\vec{x}$ can be conveniently written as
	\[
	\vec{x}^{\top}L\vec{x} = \sum_{\set{u,v}\in E(G)}(x_u - x_v)^2,
	\]
	where $x_v$ is the coordinate of $\vec{x}$ corresponding to node $v \in V$. For a vector $\vec{x}$ with $\vec{x}^T\vec{x}=1$, we therefore have $\sum_{\set{u,v}\in E(G)}(x_u - x_v)^2 \le \lambda_n$. Let $C$ be a subset of the vertices, $C \subset V$ of the graph $G$, where $|C| = k < n$. Furthermore, let $E(C)$ be the set of edges that are incident to a vertex in $C$ and a vertex in $V\setminus C$. In addition, for each vertex $v \in C$, let $x_v = w$ and for each vertex $v \in V\setminus C$ let $x_v = -y$ where $w$ and $y$ are both positive numbers. Following the above equation we have the following:
	\[
	\vec{x}^{\top}L\vec{x} = \sum_{\set{u,v}\in E(G)}(x_u - x_v)^2 = |E(C)|\cdot (w + y)^2.
	\]
	Therefore, in a random $\Delta$-regular graph, $|E(C)|(w + y)^2 \le \lambda_n \le \Delta + 3\sqrt{\Delta}$. As we will see, based on this, one can show that a random $\Delta$-regular graph does not have an orientation where every node either has outdegree at most $\rho_{1}\cdot\Delta - c\sqrt{\Delta}$ or indegree at most $\rho_{2}\cdot \Delta - c\sqrt{\Delta}$, where $\rho_{1} + \rho_{2} = \frac{1}{2}$. However, we need to show this for a $\Delta$-regular graph with girth $\Omega(\log_\Delta n)$, which is not true for a uniformly random $\Delta$-regular graph. With a large probability, such a graph however has at most $O(\sqrt{n})$ short cycles and as described in \cite{binarylabelings,Alon10}, it is possible to change $O(\sqrt{n})$ edges to turn a uniformly random $\Delta$-regular graph into a $\Delta$-regular graph of girth $\Omega(\log_\Delta n)$. In the following, we assume that $G$ is such a $\Delta$-regular graph. Clearly, when exchanging at most $O(\sqrt{n})$ edges, the size of every cut in a graph can shrink or grow by at most $O(\sqrt{n})$. We can therefore assume that $G$ is a high-girth $\Delta$-regular graph for which for some constant $\gamma>0$,
	\[
	(|E(C)|-\gamma\sqrt{n})(w + y)^2 \le \lambda_n \le \Delta + 3\sqrt{\Delta}.
	\]
	We have to choose the values $w$ and $y$ such that $k\cdot w^2  + (n-k)\cdot y^2=1$. We choose $w$ and $y$ as follows.
	\[
	w = \sqrt{\frac{n-k}{nk}}\quad\text{and}\quad y =\sqrt{\frac{k}{n(n-k)}}.
	\]
	Note that those values satisfy $k\cdot w^2  + (n-k)\cdot y^2=1$. One can show that they also maximize the value of $(w+y)^2$ subject to this constraint.
	For this choice of $w$ and $y$, we now get
	\[
	(w+y)^2 = \frac{1}{n}\left(\frac{k}{n - k} + \frac{n - k}{k}\right) + 2\sqrt{\frac{k(n - k)}{n^2k(n-k)}} = \frac{n}{k(n - k)}.
	\]
	We therefore get
	\[
	(|E(C)|-\gamma\sqrt{n})(w + y)^2  = (|E(C)|-\gamma\sqrt{n})\cdot \frac{n}{k(n-k)} \le \Delta + 3\sqrt{\Delta}.
	\]
	This leads to the following upper bound on the cut size $|E(C)|$:
	\begin{equation}\label{eq:cutupperbound1}
		|E(C)| \le \frac{k(n - k)}{n}\cdot(\Delta + 3\sqrt{\Delta}) + \gamma\sqrt{n}.
	\end{equation}
	In the following, we will assume that for some given $\rho_1,\rho_2$ such that $\rho_1+\rho_2=1/2$, it is possible to  compute an orientation in which every node either has outdegree at most $\rho_1\Delta - c\sqrt{\Delta}$ or indegree at most $\rho_2\Delta - c\sqrt{\Delta}$. We will assume that $|C|=k$ and that $C$ is the set of nodes for which the outdegree is at most $\rho_1\Delta - c\sqrt{\Delta}$ and that $V\setminus C$ is the set of nodes for which the indegree is at most $\rho_2\Delta - c\sqrt{\Delta}$. We will see that this leads to a contradiction, independently of the size $k$ of $C$.
	
	First, observe that we need $\rho_1,\rho_2 \geq c/\sqrt{\Delta}$. Assume that this is not the case and that, e.g., $\rho_1<c/\sqrt{\Delta}$. Then, the requirement becomes that all nodes either have outdegree $<0$, which is clearly impossible, or indegree $<\Delta/2$, which is not possible for all nodes. We next argue that both $k$ and $n-k$ cannot be too small. For a given orientation, let $d_{out}(v)$ be the outdegree of some node $v$. In a $\Delta$-regular graph in which $k$ nodes have outdegree at most $\rho_1\Delta$, we have
	\[
	\frac{n\Delta}{2} = \sum_{v\in V} d_{out}(v) \leq k\cdot\rho_1\Delta + (n-k)\cdot \Delta.
	\]
	This implies that
	\begin{equation}\label{eq:k-nk-lower}
		n-k \geq n\cdot\left(1 - \frac{1}{2(1-\rho_1)}\right) \stackrel{(\rho_1 \leq 1/2-c/\sqrt{\Delta})}{\geq}
		n\cdot\left(1-\frac{1}{1 + \frac{2c}{\sqrt{\Delta}}}\right) \geq n\cdot\frac{c}{\sqrt{\Delta}}.
	\end{equation}
	In the last inequality, we use that $2c\leq\sqrt{\Delta}$. Note that otherwise, the requirement that nodes must have outdegree at most $\rho_1\Delta -c\sqrt{\Delta}$ or indegree at most $\rho_2\Delta - c\sqrt{\Delta}$ becomes trivially impossible. We analogously also get the same lower bound for $k$.
	
	In the following let $G_C$ be the subgraph of $G$ induced by the nodes in $C$ and let $G_{V\setminus C}$ be the subgraph of $G$ induced by the nodes in $V\setminus C$. Because the outdegree of every nodes in $C$ is at most $\rho_1\Delta - c\sqrt{\Delta}$, the average degree in $G_C$ is at most $2(\rho_1\Delta - c\sqrt{\Delta})$. Similarly, the average degree in $G_{V\setminus C}$ is at most $2(\rho_2\Delta - c\sqrt{\Delta})$. Because every node in $C$ has degree $\Delta$, on average, the nodes in $C$ have at least $\Delta-2(\rho_1\Delta-c\sqrt{\Delta})$ edges in the cut $E(C)$ and analogously, the nodes in $V\setminus C$ on average have at least $\Delta - 2(\rho_2 -c\sqrt{\Delta})$ edges in the cut $E(C)$. We therefore have the following two lower bounds on the size of $E(C)$:
	\begin{eqnarray*}
		|E(C)| & \geq & k\cdot \big(\Delta-2(\rho_1\Delta-c\sqrt{\Delta})\big),\\
		|E(C)| & \geq & (n-k)\cdot \big(\Delta-2(\rho_2\Delta-c\sqrt{\Delta})\big).
	\end{eqnarray*}
	If we combine the two inequalities with \Cref{eq:cutupperbound1}, we get
	\begin{eqnarray*}
		\Delta-2(\rho_1\Delta-c\sqrt{\Delta}) & \leq\ \frac{|E(C)|}{k}\ \leq &  \frac{n - k}{n}\cdot(\Delta + 3\sqrt{\Delta}) + \gamma\cdot\frac{\sqrt{n}}{k},\\
		\Delta-2(\rho_2\Delta-c\sqrt{\Delta}) & \leq\ \frac{|E(C)|}{n-k}\ \leq & \frac{k}{n}\cdot(\Delta + 3\sqrt{\Delta}) + \gamma\cdot\frac{\sqrt{n}}{n-k}.
	\end{eqnarray*}
	By using the lower bound on $k$ and $n-k$ from \Cref{eq:k-nk-lower}, we get that $\gamma\sqrt{n}/k \leq \gamma\sqrt{n\Delta}/(cn) \leq \gamma/c$ and similarly that $\gamma\sqrt{n}/(n-k)\leq \gamma/c$. As we can choose the constant $c$ sufficiently large, we assume that $c\geq 2\gamma$ such that $\gamma/c\leq 1/2\leq \sqrt{\Delta}/2$. By combining the two inequalities above, we then get
	\[
	2\Delta - 2(\rho_1+\rho_2)\Delta + 2c\sqrt{\Delta} \leq \Delta + 4\sqrt{\Delta}.
	\]
	Solving for $\rho_1+\rho_2$, we get
	\[
	\rho_1 + \rho_2 \geq \frac{1}{2} + (c-2)\cdot\sqrt{\Delta} \stackrel{(\text{if }c>2)}{>} \frac{1}{2}.
	\]
	Hence, if we choose $c>2$, we get a contradiction to the assumption that $\rho_1+\rho_2=1/2$. This concludes the proof.
\end{proof}

\section{Implementation with Small Messages}\label{sec:smallmsg}
The primary focus of this paper is on exploring the randomized
complexity of LCL problems in the \LOCAL model, where messages can be
of arbitrary size. We thus do not specifically discuss algorithms for
the \CONGEST model, where messages have to consist of $O(\log n)$
bits. We however believe that none of our algorithms inherently rely on the unbounded bandwidth offered by the \LOCAL model. In fact, we expect that all our algorithms can be adapted to work in the \CONGEST model with minimal modifications. We next briefly sketch the main ideas of how to achieve this.

The degree splitting algorithm, for example, is built on a reduction to multiple instances of the sinkless orientation problem. To adapt this algorithm to the \CONGEST model, we thus need to ensure that both the reduction process and the sinkless orientation problem can be handled within the message size constraints of the \CONGEST model. We are aware of the existence of an $O(\log n)$-round deterministic and an $O(\log\log n)$-round randomized sinkless orientation algorithm for the \CONGEST model~\cite{randomizedSO}. Moreover, using the general result of \cite{MU21}, a $(\log \log n)^{O(1)}$-round randomized \CONGEST sinkless orientation algorithm is straightforward. Each step of the degree splitting algorithm operates on a virtual graph where several instances of the sinkless orientation are executed in parallel on edge-disjoint subgraphs. Since the nodes of the virtual graph are a subset of the nodes of the communication graph, and the edges consist of edge-disjoint paths, this simulation can be efficiently executed without increasing the message size. Given sinkless orientation algorithms for the \CONGEST model that are as efficient as the \LOCAL model counterparts, we therefore expect that the algorithms underlying the results of \Cref{thm:balancedsplitting,thm:arbitrary2splitting} can be transformed to work in the \CONGEST model at no additional asymptotic cost.

For the edge orientation algorithms that depend on black-box
applications of LLL algorithms, an efficient \CONGEST implementation can likely be achieved by leveraging recent generic \CONGEST algorithms for LLL problems, such as those in \cite{MU21,HMP24}.

\section{Conclusions}

In this paper, we explored the randomized complexity of a subclass of
distributed binary labeling problems in the \LOCAL model. We provided
in some cases tight and in some cases almost tight characterizations
of the randomized complexity of several natural problems, including
red-blue edge colorings and various edge orientation tasks. In
particular, we showed that certain edge coloring and orientation
problems with a deterministic complexities of $\Omega(\log n)$ can be
solved exponentially faster using randomization, achieving time
$O(\log \log n)$ in bounded-degree trees and time
$\tilde{O}(\Delta\cdot\log^4 \log n)$ in general graphs.

We hope that our paper provides a foundation for several interesting
directions for future research. One immediate avenue is to explore the
extension of our techniques to the full set of binary labeling
problems that have been studied in \cite{binarylabelings} (which
includes coloring and orientation problems in
hypergraphs). Characterizing the randomized complexities of binary
labeling problems is a natural next step since they form the only
general class of problems for which a full characterization of the
deterministic complexity is known. Additionally, degree splittings
might be useful for highlevel problems beyond standard edge coloring
for which they have so far mostly been used. We hope that these
directions will lead to deeper insights into the power of
randomization for distributed graph algorithms.
\newpage
\bibliographystyle{alpha}
\bibliography{arxiv}

\end{document}